\documentclass[11pt,a4paper]{article}
\usepackage[utf8]{inputenc}

\usepackage{amsmath}
\usepackage{amsfonts}
\usepackage{amssymb}
\usepackage{amsthm}
\usepackage{breqn}
\usepackage{latexsym}
\usepackage{physics}

\usepackage{verbatim}

\usepackage[margin=2.2cm]{geometry}
\usepackage[all,cmtip]{xy}
\usepackage{hyperref}

\theoremstyle{definition}
\newtheorem{defn}{Definition}[section]

\theoremstyle{definition}
\newtheorem{prop}{Proposition}[section]

\newtheorem{thm}{Theorem}[section]

\newtheorem{lemma}{Lemma}[section]

\theoremstyle{remark}
\newtheorem{rem}{Remark}[section]

\newtheorem{hyp}{Hypothesis}[section]

\newtheorem{corollary}{Corollary}[section]

\newcommand{\bigslant}[2]{{\raisebox{.2em}{$#1$}\left/\raisebox{-.2em}{$#2$}\right.}}

\title{Fundamental solutions and Hadamard states for a scalar field with arbitrary boundary conditions on an asymptotically AdS spacetimes}

\author{Claudio Dappiaggi$^{1,2,3,a}$, Alessio Marta$^{4,5,6,b}$,
\vspace{4mm}\\
{\small $^1$ Dipartimento di Fisica -- Universit{\`a} di Pavia, Via Bassi 6, 27100 Pavia, Italy.}\vspace{2mm}\\
{\small $^2$ INFN, Sezione di Pavia -- Via Bassi 6, 27100 Pavia, Italy.}\vspace{2mm}\\
{\small $^3$ Istituto Nazionale di Alta Matematica -- Sezione di Pavia, Via Ferrata, 5, 27100 Pavia, Italy.}\vspace{2mm}\\
{\small $^4$ Dipartimento di Matematica -- Universit{\`a} di Milano, Via Cesare Saldini, 50 --  I-20133 Milano, Italy.}\vspace{2mm}\\
{\small $^5$ INFN, Sezione di Milano -- Via Celoria, 16 -- I-20133 Milano, Italy.}\vspace{2mm}\\
{\small $^6$ Istituto Nazionale di Alta Matematica -- Sezione di Milano, Via Saldini, 50, I-20133 Milano, Italy.}\vspace{4mm}\\
{\footnotesize  ~$^a$ claudio.dappiaggi@unipv.it~,~$^b$ alessio.marta@unimi.it}
}

\date{\today}

\begin{document}

	\maketitle

\begin{abstract}
We consider the Klein-Gordon operator on an $n$-dimensional asymptotically anti-de Sitter spacetime $(M,g)$ together with arbitrary boundary conditions encoded by a self-adjoint pseudodifferential operator on $\partial M$ of order up to $2$. Using techniques from $b$-calculus and a propagation of singularities theorem, we prove that there exist advanced and retarded fundamental solutions, characterizing in addition their structural and microlocal properties. We apply this result to the problem of constructing Hadamard two-point distributions. These are bi-distributions which are weak bi-solutions of the underlying equations of motion with a prescribed form of their wavefront set and whose anti-symmetric part is proportional to the difference between the advanced and the retarded fundamental solutions. In particular, under a suitable restriction of the class of admissible boundary conditions and setting to zero the mass, we prove their existence extending to the case under scrutiny a deformation argument which is typically used on globally hyperbolic spacetimes with empty boundary.
\end{abstract}

\section{Introduction}\label{Sec: Introduction}

The $n$-dimensional anti-de Sitter spacetime (AdS$_n$) is a maximally symmetric solution of the vacuum Einstein equations with a negative cosmological constant. From a geometric viewpoint it is noteworthy since it is not globally hyperbolic and it possesses a timelike conformal boundary. Due to these features the study of hyperbolic partial differential equations on top of this background becomes particularly interesting, especially since the initial value problem does not yield a unique solution unless suitable boundary conditions are assigned. Therefore several authors have investigated the properties of the Klein-Gordon equation on an AdS spacetime, see {\it e.g.} \cite{Bachelot:2010zw,Enciso:2013lza,Hol12,War1,Vasy12} to quote some notable examples, which have inspired our analysis. 

A natural extension of the framework outlined in the previous paragraph consists of considering a more general class of geometries, namely the so-called $n$-dimensional asymptotically AdS spacetimes, which share the same behaviour of AdS$_n$ in a neighbourhood of conformal infinity. In this case the analysis of partial differential equations such as the Klein-Gordon one becomes more involved due to admissible class of backgrounds and, in particular, due to the lack of isometries of the metric. Noteworthy has been the recent analysis by Gannot and Wrochna, \cite{GaWr18}, in which, using techniques proper of $b$-calculus they have investigated the structural properties of the Klein-Gordon operator with Robin boundary conditions. In between the several results proven, we highlight in particular the theorem of propagation of singularities and the existence of advanced and retarded fundamental solutions.

Yet, as strongly advocated in \cite{DDF18}, the class of boundary conditions which are of interest in concrete models is  greater than the one considered in \cite{GaWr18}, a notable example in this direction being the so-called Wentzell boundary conditions, see {\it e.g.} \cite{Coclite:2014, Dappiaggi:2018pju, Favini:2002, Ueno:1973, Zahn:2015due}. For this reason in \cite{Dappiaggi:2020yxg}, we started an investigation aimed at generalizing the results of \cite{GaWr18} proving a theorem of propagation of singularities for the Klein-Gordon operator on an asymptotically anti-de Sitter spacetime $M$ such that the boundary condition is implemented by a $b$-pseudodifferential operator $\Theta\in\Psi^k_b(\partial M)$ with $k\leq 2$, see Section \ref{Sec: b-pseudodifferential operators} for the definitions.

Starting from this result, in this work we proceed with our investigation and, still using techniques proper of $b$-calculus, we discuss the existence of advanced and retarded fundamental solutions for the Klein-Gordon operator with prescribed boundary conditions. The first main result that we prove is the following:

\begin{thm}\label{first main result}
Let $P_\Theta$ be the Klein-Gordon operator as per Equation \eqref{Eq: PTheta} where $\Theta$ abides to Hypothesis \ref{hypothesis_WF_Theta}. Then there exist unique retarded $(+)$ and advanced $(-)$ propagators, that is continuous operators $G_{\Theta}^\pm:\dot{\mathcal{H}}^{-1,m+1}_\pm(M) \rightarrow \mathcal{H}^{1,m}_\pm(M)$ such that $P_\Theta G_{\Theta}^\pm = \mathbb{I}$ on $\dot{\mathcal{H}}^{-1,m+1}_\pm(M)$ and $G_{\Theta}^\pm P_\Theta = \mathbb{I}$ on $\mathcal{H}^{1,m}_{\pm,\Theta}(M)$. Furthermore, $G_{\Theta}^\pm$ is a continuous map from $\dot{\mathcal{H}_0^{-1,\infty}}(M)$ to $\mathcal{H}_{loc}^{1,\infty}(M)$ where the subscript $0$ indicates that we consider only functions of compact support. 
\end{thm}

Here the spaces $\dot{\mathcal{H}}^{-1,s+1}_\pm(M)$, $\mathcal{H}^{1,s}_\pm(M)$ as well as $\dot{\mathcal{H}_{0}^{-1,\infty}}(M)$, $\mathcal{H}_{loc}^{1,\infty}(M)$ and $\mathcal{H}^{1,m}_{\pm,\Theta}(M)$ are characterized in Definition \ref{Def: tildeH spaces} and in Section \ref{Sec: Fundamental Solutions}, see in particular Equations \eqref{Eq: suspsaces of H1(M) +}, \eqref{Eq: suspsaces of H1(M) -} and \eqref{Eq: timelikecompact H1(M) with boundary condition}.

In addition, we characterize the wavefront set of the advanced $(-)$ and of the retarded $(+)$ fundamental solutions as well as their wavefront set, thanks to the theorem of propagation of singularities proven in \cite{Dappiaggi:2020yxg}. This result allows us to discuss a notable application which is strongly inspired by the so-called algebraic approach to quantum field theory, see {\it e.g.} \cite{Book_AQFT} for a recent review. In this framework a key r\^{o}le is played by the so-called Hadamard two-point distributions, which are positive bi-distributions on the underlying background which are characterized by the following defining properties: they are bi-solutions of the underlying equations of motion, their antisymmetric part is proportional to the difference between the advanced and retarded fundamental solutions and their wavefront set has a prescribed form, see {\it e.g.} \cite{Khavkine:2014mta}. If the underlying background is globally hyperbolic and with empty boundary, the existence of these two-point distributions is a by-product of the standard H\"ormander propagation of singularities theorem and of a deformation argument due to Fulling, Narcovich and Wald, see \cite{FNW81}. 

In the scenarios investigated in this work this conclusion does no longer apply since we are considering asymptotically AdS spacetimes which possess in particular a conformal boundary. At the level of Hadamard two-point distributions this has long-standing consequences since even the standard form of the wavefront set has to be modified to take into account reflection of singularities at the boundary, see \cite{DF18} and Definition \ref{Def: Hadamard 2-pt function} below. Our second main result consists of showing that, under a suitable restriction on the allowed class of boundary conditions, see Hypothesis \ref{hypothesis_WF_Theta} in the main body of this work, it is possible to prove existence of Hadamard two-point distributions. First we focus on static spacetimes and, using spectral techniques, we construct explicitly an example, which, in the language of theoretical physics, is often referred to as the ground state. Subsequently we show that, starting from this datum and using the theorem of propagation of singularities proven in \cite{Dappiaggi:2020yxg}, we can use also in this framework a deformation argument to infer the existence of an Hadamard two-point distribution on a generic $n$-dimensional asymptotically AdS spacetime. It is important to observe that this result is in agreement and it complements the one obtained in \cite{Wro17}. To summarize our second main statement is the following, see also Definition \ref{Def: admissible boundary conditions} for the notion of static and of physically admissible boundary conditions:  

\begin{thm}
Let $(M,g)$ be a globally hyperbolic, asymptotically anti-de Sitter spacetime and let $(M_S,g_S)$ be its static deformation as per Lemma \ref{Lem: deformation_spacetime}. Let $\Theta_K$ be a static and physically admissible boundary condition so that the Klein-Gordon operator $P_{\Theta_K}$ on $(M_S,g_S)$ admits a Hadamard two-point function as per Proposition \ref{prop:wavefront_bulk_to_bulk}. Then there exists a Hadamard two point-function on $(M,g)$ for the associated Klein-Gordon operator with boundary condition ruled by $\Theta_K$.
\end{thm}

It is important to stress that the deformation argument forces us to restrict in the last part of the paper the class of admissible boundary conditions and notable examples such as those of Wentzell type are not included. They require a separate analysis of their own \cite{inpreparation}. 

\vskip .2cm

The paper is structured as follows. In Section \ref{Sec: Geometric Data} we recollect the main geometric data, particularly the notions of globally hyperbolic spacetime with timelike boundary and that of asymptotically AdS spacetime. In Section \ref{Sec: Analytic Preliminaries} we discuss the analytic data at the heart of our analysis. We start from a succinct review of $b$-calculus in Section \ref{Sec: b-pseudodifferential operators}, followed by one of twisted Sobolev spaces and energy forms. In Section \ref{Sec: The boundary value Problem} we formulate the dynamical problem, we are interested in, both in a strong and in a weak sense. In Section
\ref{Sec: Fundamental Solutions} we obtain our first main result, namely the existence of advanced and retarded fundamental solutions for all boundary conditions abiding to Hypothesis \ref{hypothesis_WF_Theta}. In addition we investigate the structural properties of these propagators and we characterize their wavefront set. In Section \ref{Sec: Hadamard two-point distributions} we investigate the existence of Hadamard two-point distributions in the case of vanishing mass. First, in Section \ref{Sec: Static Fundamental Solutions} and \ref{Sec: Existence of Hadamard States on Static Spacetimes}, using spectral techniques we prove their existence on static spacetimes though for a restricted class of admissible boundary conditions, see Hypothesis \ref{hypothesis_WF_Theta} and Definition \ref{Def: admissible boundary conditions}. Subsequently, in Section \ref{Sec: A deformation Argument}, we extend to the case in hand a deformation argument due to Fulling, Narcowich and Wald proving existence of Hadamard two-point distributions on a generic $n$-dimensional asymptotically AdS spacetime.

\section{Geometric Data}\label{Sec: Geometric Data}
In this section our main goal is to fix notations and conventions as well as to introduce the three main geometric data that we shall use in our analysis, namely globally hyperbolic spacetimes with timelike boundary, asymptotically anti-de Sitter spacetimes and manifolds of bounded geometry. We assume that the reader is acquainted with the basic notions of Lorentzian geometry, {\it cf.} \cite{ONeill83}. Throughout this paper with {\em spacetime}, we indicate always a smooth, connected, oriented and time oriented Lorentzian manifold $M$ of dimension $\dim M=n\geq 2$ equipped with a smooth Lorentzian manifold $g$ of signature $(-,+,\dots,+)$. With $C^\infty(M)$ ({\em resp.} $C^\infty_0(M)$) we indicate the space of smooth ({\em resp.} smooth and compactly supported) functions on $M$, while $\dot{C}^\infty(M)$ ({\em resp.} $\dot{C}^\infty_0(M)$) stands for the collection of all smooth ({\em resp.} smooth and compactly supported) functions vanishing at $\partial M$ with all their derivatives. In between all spacetimes, the following class plays a notable r\^{o}le \cite{AFS18}.

\begin{defn}\label{Def: globally hyperbolic}
	Let $(M,g)$ be a spacetime with non empty boundary $\iota:\partial M\to M$. We say that $(M,g)$ 
	\begin{enumerate}
		\item has a {\bf timelike boundary} if $(\partial M,\iota^*g)$ is a smooth, Lorentzian manifold,
		\item is {\bf globally hyperbolic} if it does not contain closed causal curves and if, for every $p,q\in M$, $J^+(p)\cap J^-(q)$ is either empty or compact. 
	\end{enumerate}
	If both conditions are met, we call $(M,g)$ a {\em globally hyperbolic spacetime with timelike boundary} and we indicate with $\mathring{M}=M\setminus \partial M$ the interior of $M$.
\end{defn}

Observe that, for simplicity, we assume throughout the paper that also $\partial M$ is {\em connected}. Notice in addition that Definition \ref{Def: globally hyperbolic} reduces to the standard notion of globally hyperbolic spacetimes when $\partial M=\emptyset$. The following theorem, proven in \cite{AFS18}, gives a more explicit characterization of the class of manifolds, we are interested in and it extends a similar theorem valid when $\partial M=\emptyset$.

\begin{thm}\label{Thm: globally hyperbolic}
	Let $(M,g)$ be an $n$-dimensional globally hyperbolic spacetime with timelike boundary. Then it is isometric to a Cartesian product $\mathbb{R}\times\Sigma$ where $\Sigma$ is an $(n-1)$-dimensional Riemannian manifold. The associated line element reads 
	\begin{equation}\label{Eq: line element}
	ds^2=-\beta d\tau^2 + \kappa_\tau,
	\end{equation}
	where $\beta\in C^\infty(\mathbb{R}\times\Sigma;(0,\infty))$ while $\tau:\mathbb{R}\times\Sigma\to\mathbb{R}$ plays the r\^{o}le of time coordinate. In addition $\mathbb{R}\ni\tau\mapsto\kappa_{\tau}$ identifies a family of Riemmannian metrics, smoothly dependent on $\tau$ and such that, calling $\Sigma_\tau\doteq\{\tau\}\times\Sigma$, each $(\Sigma_\tau,\kappa_\tau)$ is a Cauchy surface with non empty boundary.
\end{thm}

\begin{rem}\label{Rem: boundary metric}
	Observe that a notable consequence of this theorem is that, calling $\iota_{\partial M}:\partial M\to M$ the natural embedding map, then $(\partial M,h)$ where $h=\iota^*_{\partial M}g$ is a globally hyperbolic spacetime. In particular the associated line element reads 
	$$ds^2|_{\partial M}=-\beta|_{\partial M}d\tau^2+\kappa_{\tau}|_{\partial M}.$$
\end{rem}

\noindent In addition to Definition \ref{Def: globally hyperbolic} we consider another notable class of spacetimes introduced in \cite{GaWr18}. 

\begin{defn}\label{Def: asymptotically AdS}
	Let $M$ be an n-dimensional manifold with non empty boundary $\partial M$. Suppose that $\mathring{M}=M\setminus\partial M$ is equipped with a smooth Lorentzian metric $g$ and that 
	\begin{itemize}
		\item[a)] If $x \in \mathcal{C}^\infty(M)$ is a boundary function, then $\widehat{g} = x^2 g$ extends smoothly to a Lorentzian metric on $M$.
		\item[b)] The pullback $h=\iota^*_{\partial M}\widehat{g}$ via the natural embedding map $\iota_{\partial M}:\partial M\to M$ individuates a smooth Lorentzian metric.
		\item[c)] $\widehat{g}^{-1}(dx,dx)=1$ on $\partial M$.
	\end{itemize}
	Then $(M,g)$ is called an {\em asymptotically anti-de Sitter (AdS) spacetime}. In addition, if $(M,\widehat{g})$ is a globally hyperbolic spacetime with timelike boundary, {\it cf.} Definition \ref{Def: globally hyperbolic}, then we call $(M,g)$ a {\em globally hyperbolic asymptotically AdS spacetime}.
\end{defn}

Observe that conditions a), b) and c) are actually independent from the choice of the boundary function $x$ and the pullback $h$ is actually determined up to a conformal multiple since there exists always the freedom of multiplying the boundary function $x$ by any nowhere vanishing $\Omega\in C^\infty(M)$. Such freedom plays no r\^{o}le in our investigation and we shall not consider it further. Hence, for definiteness, the reader can assume that a global boundary function $x$ has been fixed once and for all.

As a direct consequence of the collar neighbourhood theorem and of the freedom in the choice of the boundary function in Definition \ref{Def: asymptotically AdS}, this can always be engineered in such a way, that, given any $p\in\partial M$, it is possible to find a neighbourhood $U\subset\partial M$ containing $p$ and $\epsilon>0$ such that on $U\times[0,\epsilon)$ the line element associated to $g$ reads
\begin{equation}\label{Eq: metric near the boundary}
ds^2 = \frac{-dx^2+h_x}{x^2}
\end{equation}
where $h_x$ is a family of Lorentzian metrics depending smoothly on $x$ such that $h_0\equiv h$.

\begin{rem}
	It is important to stress that the notion of asymptotically AdS spacetime given in Definition \ref{Def: asymptotically AdS} is actually more general than the one given in \cite{Ashtekar:1999jx}, which is more commonly used in the general relativity and theoretical physics community. Observe in particular that $h_x$ in Equation \eqref{Eq: metric near the boundary} does not need to be an Einstein metric nor $\partial M$ is required to be diffeomorphic to $\mathbb{R}\times\mathbb{S}^{n-2}$. Since we prefer to make a close connection to both \cite{GaWr18} and \cite{Dappiaggi:2020yxg} we stick to their nomenclature.
\end{rem}

\begin{rem}
	Throughout the paper, with the symbols $\tau$ and $x$ we shall always indicate respectively the time coordinate as in Equation \eqref{Eq: line element} and the spatial coordinate as in Equation \eqref{Eq: metric near the boundary}.
\end{rem}

\subsection{Manifolds of bounded geometry}\label{Sec: Manifolds of bounded geometry}

To conclude this section we introduce the manifolds of bounded geometry since they are the natural arena where one can  define Sobolev spaces when the underlying background has a non empty boundary. In this section we give a succinct survey of the main concepts and of those results which will play a key r\^{o}le in our analysis. An interested reader can find more details in \cite{Sch01,ANN16,GS13,GOW17} as well as in \cite[Sec. 2.1 \& 2.2]{DDF18}. 

\begin{defn}\label{def:manifold_bounded_wb}
	A Riemannian manifold $(N,h)$ with empty boundary is of bounded geometry if 
	\begin{itemize}
		\item[a)] The injectivity radius $r_{inj}(N)$ is strictly positive,
		\item[b)] $N$ is of totally bounded curvature, namely for all $k \in \mathbb{N}\cup\{0\}$ there exists a constant $C_k>0$ such that $\| \bigtriangledown^k R\|_{L^\infty(M)} < C_k$.
	\end{itemize}
\end{defn}

This definition cannot be applied slavishly to a manifold with non empty boundary and, to extend it, we need to introduce a preliminary concept.

\begin{defn}\label{def:submanifold_bounded}
	Let $(N,h)$ be a Riemannian manifold of bounded geometry and let $(Y,\iota_Y)$ be a codimension $k$, closed, embedded smooth submanifold with an inward pointing, unit normal vector field $\nu_Y$. The submanifold $(Y,\iota^*_Y g)$ is of bounded geometry if:
	\begin{itemize}
		\item[a)] The second fundamental form $II$ of $Y$ in $N$ and all its covariant derivatives along $Y$ are bounded,
		\item[b)] There exists $\varepsilon_Y>0$ such that the map $\phi_{\nu_Y}:Y\times(-\varepsilon_Y,\varepsilon_Y) \rightarrow N$ defined as $(x,z) \mapsto \phi_{\nu_Y}(x,z)\doteq exp_x(z \nu_{Y,x})$ is injective. 
	\end{itemize}
\end{defn} 

These last two definitions can be combined to introduce the following notable class of Riemannian manifolds

\begin{defn}\label{Def: Riemannian manifold with boundary and of bounded geometry}
	Let $(N,h)$ be a Riemannian manifold with $\partial N\neq\emptyset$. We say that $(N,h)$ is of bounded geometry if there exists a Riemannian manifold of bounded geometry $(N^\prime,h^\prime)$ of the same dimension as $N$ such that:
	\begin{itemize}
		\item[a)] $N \subset N^\prime$ and $h = h^\prime|_{N}$
		\item[b)] $(\partial N, \iota^*h^\prime)$ is a bounded geometry submanifold of $N^\prime$, where $\iota:\partial N \rightarrow N^\prime$ is the embedding map.
	\end{itemize}
\end{defn}

\begin{rem}
	Observe that Definition \ref{Def: Riemannian manifold with boundary and of bounded geometry} is independent from the choice of $N^\prime$. For completeness, we stress that an equivalent definition which does not require introducing $N^\prime$ can be formulated, see for example \cite{Sch01}.
\end{rem}

\noindent Definition \ref{Def: Riemannian manifold with boundary and of bounded geometry} applies to a Riemannian scenario, but we are particularly interested in Lorentzian manifolds. In this case the notion of bounded geometry can be introduced as discussed in \cite{GOW17} for the case of a manifold without boundary, although the extension is straightforward. More precisely let us start from $(N,h)$ a Riemannian manifold of bounded geometry such that $\dim N=n$. In addition we call $BT^m_{m^\prime}(B_n(0,\frac{r_{inj}(N)}{2}),\delta_E)$, the space of all bounded tensors on the ball $B_n(0,\frac{r_{inj}(N)}{2})$ centered at the origin of the Euclidean space $(\mathbb{R}^n,\delta_E)$ where $\delta_E$ stands for the flat metric. For every $m,m^\prime\in\mathbb{N}\cup\{0\}$, we denote with $BT^m_{m^\prime}(N)$ the space of all rank $(m,m^\prime)$ tensors $T$ on $N$ such that, for any $p\in M$, calling $T_p\doteq(\exp_p\circ e_p)^*T$ where $e_p:(\mathbb{R}^n,\delta)\to (T_pN, h_p)$ is a linear isometry,  the family $\{T_p\}_{p\in M}$ is bounded on $BT^m_{m^\prime}(B_n(0,\frac{r_{inj}(N)}{2}),\delta_E)$. 

\begin{defn}\label{Def: Lorentzian manifold of bounded geometry}
	A smooth Lorentzian manifold $(M,g)$ is of bounded geometry if there exists a Riemannian metric $\widehat{g}$ on $M$ such that:
	\begin{itemize}
		\item[a)] $(M,\widehat{g})$ is of bounded geometry.
		\item[b)] $g \in BT^0_2 (M,\widehat{g})$ and $g^{-1} \in BT^2_0(M,\widehat{g})$.
	\end{itemize}
\end{defn}

On top of a Riemannian (or of a Lorentzian) manifold of bounded geometry $(N,h)$ we can introduce $H^k(N)\equiv W^{2,k}(N)$ which is the completion of 
$$\mathcal{E}^k(N)\doteq\{f\in C^\infty(N)\;|\;f,\nabla f,\dots,(\nabla)^k f\in L^2(N)\},$$
with respect to the norm 
$$\|f\|_{W^{2,k}(N)}=\left(\sum\limits_{i=0}^k\|(\nabla)^i f\|_{L^2(N)}\right)^{\frac{1}{2}},$$
where $\nabla$ is the covariant derivative built out of the Riemannian metric $h$, while $(\nabla)^i$ indicates the $i$-th covariant derivative. This notation is employed to disambiguate with $\nabla^i=h^{ij}\nabla_j$.

\begin{rem}\label{Rem: Sobolev vs bounded geometry}
	One might wonder why the assumption of bounded geometry is necessary since it seems to play no r\^{o}le in above characterization. The reason is actually two-fold. On the one hand it is possible to give a local definition of Sobolev spaces via a suitable choice of charts, which yields in turn a global counterpart via a partition of unity argument. Such definition is a prior different from the one given above unless one assumes to work with manifolds of bounded geometry, see \cite{GS13}. In addition such alternative characterization of Sobolev spaces allows for introducing a suitable generalization to manifolds of bounded geometry of the standard Lions-Magenes trace, which will play an important r\^{o}le especially in Section \ref{Sec: Static Fundamental Solutions}.
\end{rem}

Observe that, henceforth, we shall always assume implicitly that all manifolds that we consider are of bounded geometry.

\section{Analytic Preliminaries}\label{Sec: Analytic Preliminaries}

In this section we introduce the main analytic tools that play a key r\^{o}le in our investigation. We start by recollecting the main results from \cite{Dappiaggi:2020yxg} which are, in turn, based on \cite{GaWr18} and \cite{Vasy10,Vasy12}. 

\subsection{On b-pseudodifferential operators}\label{Sec: b-pseudodifferential operators}

In the following we assume for definiteness that $(M,g)$ is a globally hyperbolic, asymptotically $AdS$ spacetime of bounded geometry as per Definition \ref{Def: asymptotically AdS} and Definition \ref{Def: Lorentzian manifold of bounded geometry}. In addition we assume that the reader is familiar with the basic ideas and tools behind $b$-geometry, first introduced by R. Melrose in \cite{Mel92}. Here we limit ourselves to fix notations and conventions, following the presentation of \cite{GMP}. 

In the following with ${}^bTM$ we indicate the {\em $b$-tangent bundle}  which is a vector bundle whose fibres are
$${}^bT_pM=\left\{ 
\begin{array}{ll}
T_p M & p\in \mathring{M}\\
\textrm{span}_{\mathbb{R}}(x\partial_x, T_p\partial M) & p \in \partial M
\end{array}
\right., $$
where $x$ is the global boundary function introduced in Definition \ref{Def: asymptotically AdS}, here promoted to coordinate. Similarly we can define per duality the {\em $b$-cotangent bundle}, ${}^bT^*M$ which is a vector bundle whose fibers are
$${}^bT^*_pM=\left\{ 
\begin{array}{ll}
T^*_p M & p\in \mathring{M}\\
\textrm{span}_{\mathbb{R}}(\frac{dx}{x}, T^*_p\partial M) & p \in \partial M
\end{array}
\right.$$
For future convenience, whenever we fix a chart $U$ of $M$ centered at a point $p\in\partial M$, we consider $(x,y_i,\xi,\eta_i)$ and $(x,y_i,\zeta,\eta_i)$, $i=1,\dots,n-1=\dim\partial M$, local coordinates respectively of $T^*M|_U$ and of ${}^bT^*M|_U$. Since we are considering globally hyperbolic spacetimes, hence endowed with a distinguished time direction $\tau$, {\it cf.} Equation \eqref{Eq: line element}, we identify implicitly $\eta_{n-1}\equiv \tau$. In addition, observe that there exists a natural projection map 
$$\pi:T^*M\to{}^b T^*M,\quad (x,y_i,\xi,\eta_i)\mapsto \pi(x,y_i,\xi,\eta_i)=(x,y_i,x\xi,\eta_i),$$
which is non-injective. This feature prompts the definition of a very important structure in our investigation, namely the {\em compressed $b$-cotangent bundle} 
\begin{equation}\label{Eq: Compressed b-cotangent bundle}
{}^b\dot{T}^*M\doteq\pi[T^*M],
\end{equation}
which is a vector sub-bundle of ${}^bT^*M$, such that ${}^b\dot{T}^*_pM\equiv T^*_pM$ whenever $p\in\mathring{M}$. The last geometric structure that we shall need in this work is the {\em b-cosphere bundle} which is realized as the quotient manifold obtained via the action of the dilation group on $T^*_bM\setminus\{0\}$, namely
\begin{equation}\label{Eq: cosphere bundle}
	{}^bS^*M\doteq\bigslant{{}^bT^*M\setminus\{0\}}{\mathbb{R}^+}.
\end{equation}
We remark that, if we consider a local chart $U\subset M$ such that $U\cap\partial M\neq\emptyset$ and the local coordinates $(x,y_i,\zeta,\eta_i)$, $i=1,\dots,n-1=\dim\partial M$, on ${}^bT^*_UM\doteq{}^bT^*M|_U$, we can build a natural counterpart on ${}^bS^*_UM$, namely $(x,y_i,\widehat{\zeta},\widehat{\eta}_i)$ where $\widehat{\zeta}=\frac{\zeta}{|\eta_{n-1}|}$ and $\widehat{\eta_i}=\frac{\eta_i}{|\eta_{n-1}|}$. On top of these geometric structures we can define two natural classes of operators.

\begin{defn}\label{Def: algebra of pseudodifferential operators}
	Let $(M,g)$ be a globally hyperbolic, asymptotically $AdS$ spacetime. We call 
\begin{itemize}
	\item $\textbf{Diff}_b(M)\doteq\bigoplus_{k=0}^\infty\textbf{Diff}^k_b(M)$ the graded, differential operator algebra generated by $\Gamma({}^bTM)$, the space of smooth section of the $b$-tangent bundle.
	\item $\Psi_b^m(M)$ the set of properly supported $b$-pseudodifferential operators ($b-\Psi$DOs) of order $m$, $m\in\mathbb{R}$.
\end{itemize}
\end{defn}

The notion of $b-\Psi$DOs is strictly intertwined with $S^m({}^bT^*M)$ the set of all symbols of order $m$ on ${}^bT^*M$ and in particular there exists a principal symbol map 
\begin{equation}\label{Eq: principal symbol map}
\sigma_{b,m}:\Psi_b^m(M)\to S^m({}^bT^*M)/S^{m-1}({}^bT^*M),\quad A\mapsto a=\sigma_{b,m}(A),
\end{equation}
which gives rise to an isomorphism 
$$\Psi_b^m(M)/\Psi_b^{m-1}(M)\simeq S^m({}^bT^*M)/S^{m-1}({}^bT^*M).$$
In addition we can endow the space of symbols $S^m({}^bT^*M)$ with a Fr\'echet topology induced by the family of seminorms
$$ \| a \|_{N} \ = \sup_{(z,k_z) \in K_i \times \mathbb{R}^n} \max_{|\alpha|+ |\gamma| \leq N} \dfrac{|\partial_z^\alpha \partial_\zeta^\gamma a(z,k_z) | }{\langle k_z \rangle^{m-|\gamma|}},$$
where $\langle k_z \rangle = (1+|k_z|^2)^{\frac{1}{2}}$, while $\{K_i\}_{i\in I}$, $I$ being an index set, is an exhaustion of $M$ by compact subsets. Hence one can endow $S^m({}^bT^*M)$ with a metric $d$ as follows
\begin{equation*}
d(a,b) = \sum_{N \in \mathbb{N}} 2^{-N} \dfrac{\|a-b\|_{N}}{1+\|a-b \|_N}.\quad\forall a,b \in S^m({}^bT^*M)
\end{equation*}
In view of these data the following definition is natural 
\begin{defn}\label{Def: bounded PsiDOs}
	A subset of $\Psi_b^m(M)$ is called {\em bounded} if such is the associated subset of $S^m({}^bT^*M)$ with respect to the Fr\'echet topology.
\end{defn}
Finally we can recall the notion of elliptic $b-\Psi$DO and of wavefront set both of a single and of a family of pseudodifferential operators, {\it cf.} \cite{Hor1}:

\begin{defn}\label{Def: ellptic PsiDO}
	A b-pseudodifferential operator $A \in \Psi^m_b(M)$ is {\em elliptic} at a point $q_0 \in \ {}^bT^*M \setminus\{0\}$ if there exists $c\in S^{-m}(^bT^*M)$ such that
	\begin{equation*}
	\sigma_{b,m}(A)\cdot c - 1 \in S^{-1}(^bT^*M),
	\end{equation*} 
	in a conic neighbourhood of $q_0$. We call $ell_b(A)$ the (conic) subset of $^bT^*M \setminus\{0\}$ in which $A$ is elliptic.
\end{defn}

\begin{defn}\label{Def: WF of PsiDO}
	For any $P \in \Psi^m_b(M)$, we say that $(z_0,k_{z_0}) \notin WF^\prime_b(P)$ if the associated symbol $p(z,k_z)$ is such that, for every multi-indices $\gamma$ and for every $N\in\mathbb{N}$, there exists a constant $C_{N,\alpha,\gamma}$ such that
	\begin{equation*}
	|\partial_z^\alpha \partial^\gamma_{k_z} p(z,k_z)| \leq C_{N,\alpha,\gamma} \langle k_z \rangle^{-N}, 
	\end{equation*}
	for $z$ in a neighbourhood of $z_0$ and $k_z$ in a conic neighbourhood of $k_{z_0}$. 
	
	Similarly, if $\mathcal{A}$ is a bounded subset of $\Psi_b^m(M)$ and $q \in {}^bT^*M$. We say that $q \not \in WF_b^\prime(\mathcal{A})$ if there exists $B \in \Psi_b(M)$, elliptic at $q$, such that $\{ BA : A \in \mathcal{A}\}$ is a bounded subset of $\Psi_b^{-\infty}(M)$.
\end{defn}

To conclude this part of the section, we stress that, in order to study the behavior of a b-pseudodifferential operator at the boundary, it is useful to introduce the notion of \textit{indicial family}, \cite{GaWr18}. Let $A \in \Psi_b^m(M)$. For a fixed boundary function $x$, {\it cf.} Definition \ref{Def: asymptotically AdS}, and for any $v \in \mathcal{C}^\infty(\partial M)$ we define the indicial family $\widehat{N}(A)(s):C^\infty(\partial M)\to C^\infty(\partial M)$ as:
\begin{equation}\label{Eq: indicial family}
	\widehat{N}(A)(s)v = x^{-is}A \left( x^{is}u \right)|_{\partial M}
\end{equation}
where $u \in \mathcal{C}^\infty(M)$ is any function such that $u|_{\partial M}=v$.

\subsection{Twisted Sobolev Spaces}\label{Sec: Twisted Sobolev Spaces}

In this section we introduce the second main analytic ingredient that we need in our investigation. To this end, once more we consider $(M,g)$ a globally hyperbolic, asymptotically $AdS$ spacetime and the associated Klein-Gordon operator $P\doteq\Box_g-m^2$, where $m^2$ plays the r\^{o}le of a mass term, while $\Box_g$ is the D'Alembert wave operator built out of the metric $g$. It is convenient to introduce the parameter 
\begin{equation}\label{Eq: nu parameter}
\nu=\frac{1}{2}\sqrt{(n-1)^2+4m^2},
\end{equation}
which is constrained to be positive. This is known in the literature as the Breitenlohner-Freedman bound \cite{BF82}. In the spirit of \cite{GaWr18} and \cite[Sec. 3.2]{Dappiaggi:2020yxg} we introduce the following, finitely generated, space of twisted differential operators
$$\textbf{Diff}^1_\nu(M)\doteq\{x^{\nu_-}Dx^{-\nu_-}\;|\;D\in\textbf{Diff}^1(M)\},$$
where $\nu_-=\frac{n-1}{2}-\nu$, $n=\dim M$. Starting from these data, and calling with $x$ and $d\mu_g$ respectively the global boundary function, {\it cf.} Definition \ref{Def: asymptotically AdS}, and the metric induced volume measure we set
 
\begin{equation}\label{Eq: Twisted Hilbert Spaces}
\mathcal{L}^2(M)\doteq L^2(M;x^2d\mu_g)\;\;\textrm{and}\;\mathcal{H}^1(M)\doteq\{u\in\mathcal{L}^2(M)\;|\;Qu\in\mathcal{L}^2(M)\;\forall Q\in\textbf{Diff}^1_\nu(M)\}.
\end{equation}

\noindent The latter is a Sobolev space if endowed with the norm
$$\| u \|^2_{\mathcal{H}^1(M)} = \| u \|^2_{\mathcal{L}^2(M)} + \sum_{i=1}^n \|Q_i u\|^2_{\mathcal{L}^2(M)},$$
where $\{Q_i\}_{i=1\dots n}$ is a generating set of $\textbf{Diff}^1_\nu(M)$. In addition we shall be using $\mathcal{L}^2_{loc}(M)$, the space of locally square integrable functions over $M$ with respect to the measure $x^2d\mu_g$ and $\dot{\mathcal{L}^2}_{loc}(M)$ the counterpart built starting from $\dot{C}^\infty(M)$ in place of $C^\infty(M)$. Starting from these spaces we can build the first order Sobolev spaces $\mathcal{H}^1_{loc}(M)$ and $\dot{\mathcal{H}}^1_{loc}(M)$ as well as their respective topological duals, $\dot{\mathcal{H}}^{-1}_{loc}(M)$ and $\mathcal{H}^{-1}_{loc}(M)$. Finally, calling $\mathcal{E}^\prime(M)$ the topological dual space of $\dot{C}^\infty(M)$, we set 
\begin{equation}\label{Eq: H10}
\mathcal{H}^1_0(M)=\mathcal{H}^1_{loc}(M)\cap\mathcal{E}^\prime(M),
\end{equation}
while, similarly, we define $\mathcal{H}^{-1}_0(M)$. 

We discuss succinctly the interactions between $\Psi^m_b(M)$ and $\textbf{Diff}_\nu^1(M)$. We begin by studying the action of pseudodifferential operators of order zero on the spaces $\mathcal{H}^k_{loc/0}(M)$, $k=\pm 1$, just defined. Every $A \in \Psi^0_b(M)$ is a bounded operator thereon, as stated in the following lemma.

\begin{lemma}[\cite{GaWr18}, Lemma 3.8 and \cite{Vasy08}, Lemma 3.2, Corollary 3.4]\label{Lemma: PDO zero order bounded}
Let $A \in \Psi^0_b(M)$. Then $A$ is a continuous linear map
$$ \mathcal{H}^1_{loc/0}(M) \rightarrow \mathcal{H}^1_{loc/0}(M), \ \ \ \dot{\mathcal{H}}^1_{loc/0}(M) \rightarrow \dot{\mathcal{H}}^1_{loc/0}(M), $$
which extends per duality to a continuous map
$$ \dot{\mathcal{H}}^{-1}_{0/loc}(M) \rightarrow \dot{\mathcal{H}}^{-1}_{0/loc}(M), \ \ \ \mathcal{H}^{-1}_{0/loc}(M) \rightarrow \mathcal{H}^{-1}_{0/loc}(M).$$ 
\end{lemma}
The proof of this lemma gives a useful information. Let $A \in \Psi^0_b(M)$ be with compact support $U \subset M$. Then there exists $\chi \in \mathcal{C}_0^\infty(U)$ such that
\begin{equation}\label{Prop: Bound_PSI_Zero}
\| A u \|_{\mathcal{H}^k(M)} \leq C \|\chi u \|_{\mathcal{H}^k(M)},
\end{equation}
for every $u \in \mathcal{H}^k(M)$ with $k=\pm 1$. The constant $C$ is bounded by a seminorm of $A$.

To study in full generality the interactions between $\Psi^m_b(M)$ and $\textbf{Diff}_\nu^1(M)$, we need to introduce one last class of relevant spaces

\begin{defn}\label{Def: tildeH spaces}
	Let $k=-1,0,1$ and let $m \geq 0$. Given $u \in \mathcal{H}_{loc}^k(M)$ ({\em resp.} $\mathcal{H}^k(M)$), we say that $u \in \mathcal{H}_{loc}^{k,m}(M)$ ({\em resp.} $\mathcal{H}^{k,m}(M)$) if $Au \in \mathcal{H}_{loc}^k(M)$ ({\em resp.} $\mathcal{H}^k(M)$) for all $A \in \Psi^m_b(M)$. Furthermore, we define $\mathcal{H}^{k,\infty}(M)$ as:
	\begin{equation}
	\mathcal{H}^{k,\infty}(M) \doteq  \bigcap_{m=0}^\infty \mathcal{H}^{k,m}(M).
	\end{equation}
\end{defn}

\begin{rem}\label{rem: H1k_elliptic}
	As observed in \cite{Vasy08}, whenever $m$ is finite, it is enough to check that both $u$ and $Au$ lie in $\mathcal{H}^k_{loc}(M)$ for a single elliptic operator $A \in \Psi^m_b(M)$.
\end{rem}

\noindent Observe that, in full analogy to Definition \ref{Def: tildeH spaces}, we define similarly $\mathcal{H}^{k,m}_0(M)$ and $\dot{\mathcal{H}}^{k,m}_{loc}(M)$. In the following definition, we extend the notion of wavefront set to the spaces $\mathcal{H}_{loc}^{k,m}(M)$.

\begin{defn}\label{Def: wavefrontset for Hkm-loc}
	Let $k=0,\pm 1$ and let $u \in \mathcal{H}^{k,m}_{loc}(M)$, $m \in \mathbb{R}$. Given $q \in {}^bT^*M \setminus\{0\} $, we say that $q \not \in WF_b^{k,m}(u)$ if there exists $A \in \Psi_b^m(M)$ such that $q \in ell_b(A)$ and $Au \in \mathcal{H}^k_{loc}(M)$, where $ell_b$ stands for the elliptic set as per Definition \ref{Def: ellptic PsiDO}. When $m=+\infty$, we say that $q \not \in WF_b^{k,\infty}(M)$ if there exists $A \in \Psi^0_b(M)$ such that $q \in ell_b(A)$ and $Au \in \mathcal{H}^{k,\infty}_{loc}(M)$.
\end{defn}

With all these data, we can define two notable trace maps which will be a key ingredient in the next section. The following proposition summarizes the content of \cite[Lemma 3.3]{GaWr18} and \cite[Lemma 3.4]{Dappiaggi:2020yxg}:

\begin{thm}\label{Thm: gamma-}
	Let $(M,g)$ be a globally hyperbolic, asymptotically $AdS$ spacetime of bounded geometry with $n=\dim M$ and let $\nu>0$, {\it cf.} Equation \eqref{Eq: nu parameter}. Then there exists a continuous map $\widetilde{\gamma}_-:\mathcal{H}^1_0(M)\to \mathcal{H}^\nu(\partial M)$, which can be extended to a continuous map 
	$$\gamma_-:\mathcal{H}^{1,m}_{loc}(M)\to\mathcal{H}^{\nu+m}_{loc}(\partial M),\quad\forall m\leq 0.$$
\end{thm}

\begin{rem}\label{Rem: significance of gamma-}
	In order to better grasp the r\^{o}le of the trace map defined in Theorem \ref{Thm: gamma-}, it is convenient to focus the attention on $\mathbb{R}^n_+\doteq [0,\infty)\times\mathbb{R}^{n-1}$. In this setting, any $u\in\mathcal{H}^1(\mathbb{R}^n_+)$ can be restricted to the subset $[0,\epsilon)\times\mathbb{R}^{n-1}$, $\epsilon>0$ admitting an asymptotic expansion $u=x^{\nu_-}u_-+x^{r+1}u_0$ where $2r=n-2$, while $u_-\in \mathcal{H}^\nu(\mathbb{R}^n)$ and $u_0\in\mathcal{H}^1([0,\epsilon);L^2(\mathbb{R}^{n-1}))$. In this context it holds that 
	$\gamma_-(u)=u_-$.
\end{rem}
At last we recall from \cite{GaWr18} a notable property of the trace $\gamma_-$ related to its boundedness. Let $u \in \mathcal{H}(M)$, then for every $\varepsilon >0 $ there exists $C_\varepsilon >0 $ such that 
\begin{equation}\label{Prop: Bound_Gamma_Minus}
\| \gamma_- u \|^2_{L^2(\partial M)} \leq \varepsilon \| u \|^2_{\mathcal{H}^1(M)} + C_\varepsilon \| u \|^2_{\mathcal{L}^2(M)}.
\end{equation}

\subsection{Twisted Energy Form}\label{Sec: Twisted Energy Form}

In this section we focus the attention on discussing the last two preparatory key concepts before stating the boundary value problem, we are interested in. We recall that $P=\Box_g-m^2$ is the Klein-Gordon operator and, following \cite{GaWr18}, we can individuate a distinguished class of spaces whose elements enjoy additional regularity with respect to $P$:

\begin{defn}\label{Def: chik spaces}
	Let $(M,g)$ be a globally hyperbolic, asymptotically anti-de Sitter spacetime and let $P$ be the Klein-Gordon operator. For all $m \in \mathbb{R}\cup\{\pm\infty\}$, we define the Frech\'et spaces 
	\begin{equation}\label{Eq: spazi chik}
	\mathcal{X}^m(M) = \{u \in \mathcal{H}^{1,m}_{loc}(M)\; |\; Pu \in x^2 \mathcal{H}^{0,m}_{loc}(M) \},
	\end{equation}
	with respect to the seminorms
	\begin{equation}\label{Eq: seminorme chik}
	\norm{u}_{\mathcal{X}^m(M)} = \norm{\phi u}_{\mathcal{H}^{1,m}(M)}+\norm{x^{-2}\phi P u}_{\mathcal{H}^{0,m}(M)},
	\end{equation}
	where $\phi\in C^\infty_0(M)$.
\end{defn}

At this point we are ready to introduce a suitable energy form. The standard definition must be adapted to the case in hand, in order to avoid divergences due to the behaviour of the solutions of the Klein-Gordon equation at the boundary. To this end it is convenient to make use of the so-called {\em admissible twisting functions}, that is, calling $x$ the global boundary function as per Definition \ref{Def: asymptotically AdS}, the collection of $F\in x^{\nu_-}C^\infty(M)$ such that
\begin{enumerate}
	\item $x^{-\nu_-}F>0$ on $M$,
	\item $S_F\doteq F^{-1}P(F)\in x^2 L^\infty(M)$ where $P$ is the Klein-Gordon operator.
\end{enumerate}
For any such function, we can define a twisted differential
\begin{equation}\label{Eq: Twisted Differential}
d_F\doteq F\circ d\circ F^{-1}:\dot{C}^\infty(M)\to \dot{C}^\infty(M;T^*M),\quad v\mapsto d_F(v)=dv+v F^{-1}(dF).
\end{equation}
Accordingly we can introduce the {\em twisted Dirichlet (energy) form}
\begin{equation}\label{Eq: Twisted Dirichlet Form}
\mathcal{E}_0(u,v)\doteq -\int\limits_M g(d_Fu, d_F\overline{v})d\mu_g.\quad\forall u,v,\in\mathcal{L}^2_{loc}(M)
\end{equation}
Starting from these data, we are ready to introduce a second trace map. More precisely we start from 
$$\widetilde{\gamma}_+:\mathcal{X}^\infty(M)\to\mathcal{H}^{1,\infty}_{loc}(M)\quad u\mapsto \widetilde{\gamma}_+(u)=x^{1-2\nu}\partial_x(F^{-1}u)|_{\partial M}.$$
Calling $d^\dagger_F$ the formal adjoint of $d_F$ as in Equation \eqref{Eq: Twisted Differential} with respect to the inner product on $L^2(M;d\mu_g)$ we observe that, on account of the identity $P=-d^\dagger_F d_F+F^{-1}P(F)$, the following Green's formula holds true for all $u\in\mathcal{X}^\infty(M)$ and for all $v\in\mathcal{H}^1_0(M)$:
\begin{equation}\label{Eq: Green Formula}
\int Pu \cdot \overline{v} \ d\mu_g = \mathcal{E}_0(u,v) + \int S_Fu \cdot \bar{v} \ d\mu_g + \int \gamma_+u \cdot \gamma_-\bar{v} \ d\mu_h.
\end{equation}
With these premises the following result holds true, {\it cf.} \cite[Lemma 4.8]{GaWr18}:

\begin{lemma}\label{Lem: gamma+}
	The map $\widetilde{\gamma}_+$ can be extended to a bounded linear map 
	$$\gamma_+:\mathcal{X}^k(M)\to\mathcal{H}^{k-\nu}_{loc}(\partial M),\quad\forall k\in\mathbb{R}$$
	and, if $u \in \mathcal{X}^k(M)$, the Green's formula \eqref{Eq: Green Formula} holds true for every $v \in \mathcal{H}^{1,-k}_{0}(M)$.
\end{lemma}

\begin{rem}
	In order to better grasp the r\^{o}le of the second trace map characterized in Lemma \ref{Lem: gamma+}, it is convenient to focus once more the attention on $\mathbb{R}^n_+\doteq[0,\infty)\times\mathbb{R}^{n-1}$ endowed with a metric whose line element reads in standard Cartesian coordinates
	$$ds^2=\frac{-dx^2+h_{ab}dy^a dy^b}{x^2},$$
	where $h$ is a smooth Lorentzian metric on $\mathbb{R}^{n-1}$. Consider an admissible twisting function $F$ such that $\lim\limits_{x\to 0^+}x^{-\nu_-}F=1$ and $u\in\mathcal{H}^{1,k}_0(\mathbb{R}^n_+)$ such that $Pu\in x^2\mathcal{H}^{0,k}_0(\mathbb{R}^n_+)$ for $k\geq 0$. Then, for every $\epsilon>0$, the restriction of $u$ to $[0,\epsilon)\times\mathbb{R}^n$ admits an asymptotic expansion of the form $u=Fu_-+x^{\nu_+}u_++x^{r+2}\mathcal{H}_b^{k+2}([0,\epsilon);\mathcal{H}^{k-3}(\mathbb{R}^{n-1}))$ where $2r=n-2$ while $u_-\in\mathcal{H}^{\nu+k}(\mathbb{R}^{n-1})$ and $u_+\in\mathcal{H}^{k-1-2\nu}(\mathbb{R}^{n-1})$. In this context it holds that $\gamma_+(u)=2\nu u_+$.
\end{rem}

\subsection{The boundary value problem}\label{Sec: The boundary value Problem}

In this section we use the ingredients introduced in the previous analysis to formulate the dynamical problem we are interested in. At a formal level we look for $u\in \mathcal{H}^1_{loc}(M)$ such that 
\begin{equation}\label{Eq: Strong boundary KG}
\left\{ 
\begin{array}{l}
Pu = (\Box_g-m^2) u = f\\
\gamma_+ u = \Theta \gamma_- u
\end{array}
\right. ,
\end{equation}
where $\Theta\in\Psi^k_b(\partial M)$ while $\gamma_-,\gamma_+$ are the trace maps introduced in Theorem \ref{Thm: gamma-} and in Lemma \ref{Lem: gamma+} respectively. It is not convenient to look for strong solutions of Equation \eqref{Eq: Strong boundary KG}. More precisely, for any $\Theta\in\Psi^k_b(\partial M)$ , we assume that there exists an admissible twisting function $F$ and we define the energy functional 
\begin{equation}\label{Eq: Theta Twisted Energy Functional}
\mathcal{E}_\Theta(u,v)=\mathcal{E}_0(u,v)+\int\limits_M S_Fu\cdot\overline{v}\,d\mu_g+\int\limits_{\partial M}\Theta\gamma_-u\cdot\gamma_-\overline{v},
\end{equation}
where $S_F=F^{-1}P(F)$, $\mathcal{E}_0$ is the twisted Dirichlet form, {\it cf.} Equation \eqref{Eq: Twisted Dirichlet Form}, $u\in\mathcal{H}^1_{loc}(M)$, while $v\in\mathcal{H}^1_0(M)$. Hence, we can introduce $P_\Theta : \mathcal{H}^1_{loc}(M) \rightarrow \dot{\mathcal{H}}^{-1}_{loc}(M)$ by
\begin{equation}\label{Eq: PTheta}
\langle P_\Theta u, v \rangle = \mathcal{E}_\Theta(u,v).
\end{equation}
Observe that, on account of the regularity of $\gamma_-u$, we can extend $P_\Theta$ as an operator $P_\Theta: \mathcal{H}^{1,m}_{loc}(M) \rightarrow \dot{\mathcal{H}}^{-1,m}_{loc}(M)$, $m\in\mathbb{R}$ \cite{GaWr18}.

\begin{rem}\label{Rem: Where is gamma+}
	The reader might be surprised by the absence of $\gamma_+$ in the weak formulation of the boundary value problem as per Equation \eqref{Eq: PTheta}. This is only apparent since the last term in the right hand side of Equation \eqref{Eq: PTheta} is a by-product of the Green's formula as per Equation \eqref{Eq: Green Formula} together with the boundary condition introduced in Equation \eqref{Eq: Strong boundary KG}.
\end{rem}

We are now in position to recollect the two main results proved in \cite{Dappiaggi:2020yxg} concerning a propagation of singularities theorem for the Klein-Gordon operator with boundary conditions ruled by a pseudo-differential operator $\Theta\in\Psi^k_b(\partial M)$ with $k\leq 2$. As a preliminary step, we introduce two notable geometric structures. More precisely, since the principal symbol of $x^{-2}P$ reads $\widehat{p}\doteq\widehat{g}(X,X)$, where $X\in\Gamma(T^*M)$, the associated {\em characteristic set} is 
\begin{equation}\label{Eq: characteristic set}
\mathcal{N}=\left\{(q,k_q)\in T^*M\setminus\{0\}\;|\; \widehat{g}^{ij}(k_q)_i (k_q)_j=0\right\},
\end{equation}
while the {\em compressed characteristic set} is 
\begin{equation}\label{Eq: compressed characteristic set}
\dot{\mathcal{N}}=\pi[\mathcal{N}]\subset{}^b\dot{T}(M),
\end{equation}
where $\pi$ is the projection map from $T^*M$ to the compressed cotangent bundle, {\it cf.} Equation \eqref{Eq: Compressed b-cotangent bundle}. A related concept is the following: 

\begin{defn}\label{Def: generalized broken bicharacteristics}
	Let $I \subset \mathbb{R}$ be an interval. A continuous map $\gamma : I \rightarrow \dot{\mathcal{N}}$ is a {\em generalized broken bicharacteristic} (GBB) if for every $s_0 \in I$ the following conditions hold true:
	\begin{itemize}
		\item[a)] If $q_0 = \gamma(s_0) \in \mathcal{G}$, then for every $\omega\in \Gamma^\infty(^bT^*M)$,
		\begin{equation}
		\frac{d}{ds}(\omega \circ \gamma) = \{ \widehat{p},\pi^* \omega \}(\eta_0),
		\end{equation}
		where $\eta_0 \in \mathcal{N}$ is the unique point for which $\pi(\eta_0)=q_0$, while $\pi:T^*M\to{}^bT^*M$ and $\{,\}$ are the Poisson brackets on $T^*M$.
		\item[b)] If $q_0 = \gamma(s_0) \in \mathcal{H}$, then there exists $\varepsilon > 0$ such that $0 < |s-s_0| < \varepsilon $ implies $x(\gamma(s))\neq 0$, where $x$ is the global boundary function, {\it cf.} Definition \ref{Def: asymptotically AdS}.
	\end{itemize}
\end{defn}

With these structures and recalling in particular the wavefront set introduced in Definition \ref{Def: wavefrontset for Hkm-loc} we can state the following two theorems, whose proof can be found in \cite{Dappiaggi:2020yxg}:

\begin{thm}\label{Thm: main theorem k positivo}
	Let $\Theta \in \Psi_b^k(\partial M)$ with $0<k\leq 2$. If $u \in \mathcal{H}_{loc}^{1,m}(M)$ for  $m \leq 0$ and $s \in \mathbb{R} \cup \{ + \infty \}$, then $WF_b^{1,s}(u) \setminus \left( WF_b^{-1,s+1}(P_\Theta u) \cup WF_b^{-1,s+1}(\Theta u) \right) $ is the union of maximally extended generalized broken bicharacteristics within the compressed characteristic set $\dot{\mathcal{N}}$.
\end{thm}

\noindent In full analogy it holds also

\begin{thm}\label{Thm: main theorem k negativo}
	Let $\Theta \in \Psi_b^k(M)$ with $k \leq 0$. If $u \in \mathcal{H}_{loc}^{1,m}(M)$ for $m \leq 0$ and $s \in \mathbb{R} \cup \{ + \infty \}$, then it holds that $WF_b^{1,s}(u) \setminus WF_b^{-1,s+1}(P_\Theta u)$ is the union of maximally extended GBBs within the compressed characteristic set $\dot{\mathcal{N}}$.
\end{thm}

\section{Fundamental Solutions}\label{Sec: Fundamental Solutions}

In this section we prove the first of the main results of our work. We start by investigating the existence of fundamental solutions associated to the boundary value problem as in Equation \eqref{Eq: Strong boundary KG}. We shall uncover that a positive answer can be found, though we need to restrict suitably the class of admissible b-$\Psi$DOs $\Theta\in\Psi_b^k(\partial M)$ in comparison to that of Theorem \ref{Thm: main theorem k positivo} and \ref{Thm: main theorem k negativo}. We stress that, from the viewpoint of applications, these additional conditions play a mild r\^{o}le since all scenarios of interest are included in our analysis. 

We recall that the case of Dirichlet boundary condition was already analysed in \cite{Vasy12}, while the generalization to Robin boundary conditions was studied in \cite{War1} and \cite{GaWr18}, that we follow closely. We introduce a cutoff function playing an important r\^{o}le in the following theorems. Consider 
\begin{equation*}
\chi_0(s) =
\begin{cases}
exp(s^{-1}) \ if \ s > 0\\
0 \ \ \ \ \ \ \ \ \ \ \ if \ s \leq 0
\end{cases},
\end{equation*}
and let $\chi_1\in C^\infty(\mathbb{R})$ be such that $\chi_1(s)=0$ for all $s\in (-\infty,0]$ while $\chi_1(s)=1$ if $s\in [1,+\infty)$. For any but fixed $\tau_0,\tau_1\in\mathbb{R}$ with $\tau_0<\tau_1$, we call $\chi : (\tau_0,\tau_1) \rightarrow \mathbb{R}$ the smooth function
\begin{equation}
\chi(s)\doteq \chi_0( - \delta^{-1}(s-\tau_1)) \chi_1((s-\tau_0)/\varepsilon),
\end{equation}
where $\delta\gg 1$ while $\varepsilon\in (0,\tau_1-\tau_0)$. Under these hypotheses, calling $\chi^\prime_0=\frac{d\chi_0}{ds}$, it holds that, {\it cf.} \cite{Vasy12} 
\begin{equation} \label{Eq: chi bound}
\chi \leq -\delta^{-1} (\tau_1-\tau_0)^2 \chi^\prime\;\;\textrm{with}\;\;\chi^\prime = -\delta^{-1}\chi_0^\prime(-\delta^{-1}(s-\tau_1)).
\end{equation}
Consider $u_{loc} \in \mathcal{H}^{1,1}(M)$ such that its support lies in $[\tau_0+\varepsilon, \tau_1]\times\Sigma$, {\it cf.} Definition \ref{Def: globally hyperbolic}. As discussed in \cite{GaWr18}, one can use the cutoff function introduced to prove a twisted version of the Poincar\'e inequality proved in \cite[Proposition 2.5]{Vasy12}: 
\begin{equation}\label{Eq: Poincare}
\|(-\chi^\prime)^{1/2} u \|^2_{\mathcal{L}^2(M)} \leq C \|(-\chi^\prime)^{1/2} d_F u \|^2_{\mathcal{L}^2(M)},
\end{equation}
where $d_F$ is the twisted differential as per Equation \eqref{Eq: Twisted Differential}. 

Since we deal with a larger class of boundary conditions than those considered in \cite{Vasy12} and in \cite{GaWr18}, we need to make an additional hypothesis. Recall that, as in the previous sections, we are identifying a pseudodifferential operator on $\partial M$ with its natural extension on $M$, {\it i.e.} constant in $x$, the global boundary function.
As starting point we need a preliminary definition:
\begin{defn}\label{Def: pseudo local in time}
	Let $\Theta\in\Psi^k_b(M)$. We call it {\em local in time} if, for every $u$ in the domain of $\Theta$, $\tau(\textrm{supp}(\Theta u))\subseteq\tau(\textrm{supp}(u))$ where $\tau:\mathbb{R}\times\Sigma\to\mathbb{R}$ is the time coordinate individuated in Theorem \ref{Thm: globally hyperbolic}.
\end{defn}

Recalling \cite[Sec. 6]{Jos99} for the definition of the adjoint of a pseudodifferential operator, we can now formulate the following hypothesis 

\begin{hyp}\label{hypothesis_WF_Theta}
We consider $\Theta \in \Psi^k_b(M)$ with $k \leq 2$, only if it is local in time, see Definition \ref{Def: pseudo local in time}, and if $\Theta=\Theta^*$. 
\end{hyp}

The next step in the analysis of the problem in hand lies in proving the following lemma which generalizes a counterpart discussed in \cite{GaWr18} for the case of Robin boundary conditions. 

\begin{lemma}\label{Lemma: Bound u}
Let $u \in \mathcal{H}^{1,1}_{loc}(M)$ and let $\Theta\in\Psi^k_b(\partial M)$ be such that its canonical extension to $M$ abides to the Hypothesis \ref{hypothesis_WF_Theta}. Then there exists a compact subset $K \subset M$ and a real positive constant $C$ such that 
\begin{equation*}
\| (-\phi^\prime)^{1/2} u \|_{\mathcal{H}^1(K)} \leq C \| P_\Theta u \|_{\mathcal{H}^{-1,1}(K)},
\end{equation*}
where $\phi=\tau\chi$, $\chi$ being the same as in Equation \eqref{Eq: Poincare}, while $P_\Theta$ is defined in Equation \eqref{Eq: PTheta}.
\end{lemma}
\begin{proof}
The proof is a generalization of those in \cite{Vasy12} and \cite{GaWr18} to the case of boundary conditions encoded by pseudodifferential operators. Therefore we shall sketch the common part of the proof, focusing on the terms introduced by the boundary conditions. Adopting the same conventions as at the beginning of the section, assume that $supp(u) \subset [\tau_0+\varepsilon, \tau_1]\times\Sigma$. We start by computing a twisted version of the energy form considered in \cite{Vasy12}. Consider $\langle -i [(V^\prime)^* P_\Theta - P_\Theta V^\prime]u,u \rangle$, with $V^\prime = F V F^{-1} \in \textit{Diff}_b^{ \ 1}(M)$ and $V \in \mathcal{V}_b(M)$ with compact support. Note that, since $\Theta$ is  self-adjoint, {\it i.e.}, $\Theta=\Theta^*$, then $i [(V^\prime)^* P_\Theta - P_\Theta V^\prime]$ is a second order formally self-adjoint operator, the purpose of $V^{\prime*}$ being to remove zeroth order terms. Let $V= - \phi W$ with $W=\bigtriangledown_{\widehat{g}}\tau$. It belongs to $\mathcal{V}_b(X)$ because $\widehat{g}(dx,dt)=0$. A direct computation shows that 
\begin{equation}\label{Eq: Commutator P_Theta V^prime}
\begin{split}
\langle -i [(V^\prime)^* P_\Theta - P_\Theta V^\prime]u,u \rangle = 2 Re \langle P_\Theta u, V^\prime u \rangle = \\
=  2 Re \mathcal{E}_0(u, V^\prime u)+ 2 Re \langle S_F u, V^\prime u \rangle + 2 Re \langle \Theta \gamma_- u, \gamma_- V^\prime u\rangle,
\end{split}
\end{equation}
where $\mathcal{E}_0$ is the twisted Dirichlet energy form, {\it cf.} Equation \eqref{Eq: Twisted Dirichlet Form}, $S_F$ is defined in Section \ref{Sec: Twisted Energy Form}, while $\gamma_+$ and $\gamma_-$ are the trace maps introduced in Theorem \ref{Thm: gamma-} and in Lemma \ref{Lem: gamma+}. We analyze each term in the above sum separately. Starting form the first one and proceeding as in \cite{GaWr18}, we rewrite
\begin{equation*}
2 Re \mathcal{E}_0(u, V^\prime u) = \langle B^{ij} Q_i u, Q_j u\rangle,
\end{equation*}
where $Q_i$, $i=1,\dots,n$ is a generating set of $\textbf{Diff}^1_\nu(M)$, while the symmetric tensor $B$ is 
\begin{equation}\label{eq-tensor_B}
\begin{split}
B = - (\phi\cdot div_{\widehat{g}}W + 2 F\phi V(F^{-1})+ (n-2)\phi x^{-1}W(x) ) \widehat{g}^{-1} + \\
+ \phi\mathcal{L}_W \widehat{g}^{-1} + 2 T(W,\bigtriangledown_{\widehat{g}} \phi).
\end{split}
\end{equation}
Here $T(W,\bigtriangledown_{\widehat{g}} \phi)$ is the stress-energy tensor, with respect to $\widehat{g}$, see Definition \ref{Def: asymptotically AdS}, of a scalar field associated with $W$ and $\bigtriangledown_{\widehat{g}} \phi$, that is, denoting with $\odot$ the symmetric tensor product,
\begin{equation}
T(W,\bigtriangledown_{\widehat{g}} \phi) = ( \bigtriangledown_{\widehat{g}} \phi ) \odot W  - \frac{1}{2}\widehat{g}(\bigtriangledown_{\widehat{g}} \phi,W)\cdot \widehat{g}^{-1}.
\end{equation}
Focusing on this term and using that $\bigtriangledown_{\widehat{g}} \phi = \chi^\prime \bigtriangledown_{\widehat{g}} \tau$, a direct computation yields:
\begin{equation}
T_{\widehat{g}}(W,\bigtriangledown_{\widehat{g}} \phi) = \frac{1}{2} (\chi^\prime \circ\tau) \big[2 (\bigtriangledown_{\widehat{g}}\tau) \otimes (\bigtriangledown_{\widehat{g}}\tau) - \widehat{g}(\bigtriangledown_{\widehat{g}}\tau,\bigtriangledown_{\widehat{g}}\tau)\cdot \widehat{g}^{-1}\big].
\end{equation}
Since $\bigtriangledown_{\widehat{g}} \phi$ and $\bigtriangledown_{\widehat{g}}\tau$ are respectively past- and future-pointing timelike vectors, then $T_{\widehat{g}}(W,\bigtriangledown_{\widehat{g}} \phi)$ is negative definite.
Hence we can rewrite Equation \eqref{Eq: Commutator P_Theta V^prime} as
\begin{equation}\label{Eq:Stress-Energy=Commutator}
\begin{split}
\langle -T^{ij}_{\widehat{g}}(W,\bigtriangledown_{\widehat{g}} \phi) Q_i u, Q_j u \rangle = 
\langle -i [(V^\prime)^* P_\Theta - P_\Theta V^\prime]u,u \rangle + 2 Re \mathcal{E}_0(K^{ij} Q_i u,Q_j u)+ \\
+ 2 Re \langle S_F u, V^\prime u \rangle + 2 Re \langle \Theta \gamma_- u, \gamma_- V^\prime u\rangle,
\end{split}
\end{equation}
with
\begin{equation*}
K = - (F\phi V(F^{-1})+ (n-2)\phi x^{-1}W(x) ) \widehat{g}^{-1} + \phi\mathcal{L}_W \widehat{g}^{-1}.
\end{equation*}
Since $-T_{\widehat{g}}(W,\bigtriangledown_{\widehat{g}} \phi)^{ij}$ is positive definite, then $\mathcal{Q}(u,u)\doteq\langle -T_{\widehat{g}}(W,\bigtriangledown_{\widehat{g}} \phi)^{ij} Q_i u, Q_j u \rangle\geq 0$. This can be seen by direct inspection from the explicit form 
\begin{equation}\label{Eq: Dirichlet Form H}
\begin{split}
\mathcal{Q}(u,u) = \int_M \phi^\prime \left( (\bigtriangledown_{\widehat{g}}\tau)^i (\bigtriangledown_{\widehat{g}}\tau)^j - \dfrac{1}{2} \widehat{g}((\bigtriangledown_{\widehat{g}}\tau)^i(\bigtriangledown_{\widehat{g}}\tau)^j)\right) Q_i u \ \overline{Q_j u} \ x^2 d\mu_g \\ 
= \int_M  H( (-\phi^\prime)^{\/2} d_F u, (-\phi^\prime)^{1/2} d_F \overline{u}) x^2 d\mu_g,
\end{split}
\end{equation}
where $H$ is the sesquilinear pairing between $1$-forms induced by the metric. Focusing then on the term $\langle K^{ij} Q_i u, Q_j u \rangle$, we observe that, as a consequence of our choice for the functions $f$ and $W$, we have  $V(x) = \widehat{g}(\bigtriangledown_{\widehat{g}}\tau, \bigtriangledown_{\widehat{g}} x) = 0$ on $\partial M$. In addition it holds that $x^{-1}W(x)= \mathcal{O}(1)$ near $\partial M$, and $\mathcal{L}_V \widehat{g}^{-1} = 2 \bigtriangledown_{\widehat{g}} ( \bigtriangledown_{\widehat{g}}\tau)= 2 \widehat{\Gamma}^{i}_{\tau\tau}\partial_i$.
These observations allow us to establish the following bound, {\it cf.} \cite{Vasy12} and \cite{GaWr18}:
\begin{equation}\label{Eq: Bound Rij}
|\langle K^{ij} Q_i u, Q_j u \rangle| \leq C \| \phi^{1/2} d_F u  \|_{\mathcal{L}^2(M)} \leq C \delta^{-1} (\tau_1-\tau_0)^2 \| (-\phi^\prime)^{1/2} d_F u \|^2_{\mathcal{L}^2(M)},
\end{equation}
with $C$ a suitable, positive constant. Now we focus on establishing a bound for the terms on the right hand side of Equation \eqref{Eq:Stress-Energy=Commutator}.
We estimate the first one as follows:
\begin{gather}\label{Eq: Commutator Bound}
|\langle -i [(V^\prime)^* P_\Theta - P_\Theta V^\prime]u,u \rangle| \leq\notag \\ 
C \left( \| \phi^{1/2} FWF^{-1}  P_\Theta u \|_{\dot{\mathcal{H}}^{-1}(M)}^2 + \|\phi^{1/2} u \|_{\mathcal{H}^1(M)}^2 \right) 
+ C \left( \| \phi^{1/2} P_\Theta u \|_{\mathcal{L}^{2}(M)}^2 + \|\phi^{1/2}  FWF^{-1} u \|_{\mathcal{L}^2(M)}^2 \right) \leq\notag \\ 
\leq C \big( \| FWF^{-1} P_\Theta u \|_{\dot{\mathcal{H}}^{-1}(M)}^2 + \delta^{-1}(\tau_1-\tau_0)^2\| (-\phi^\prime)^{1/2} u \|_{\mathcal{H}^1(M)}^2 +\notag\\
+\| P_\Theta u \|_{\mathcal{L}^{2}(M)}^2 + \delta^{-1}(\tau_1-\tau_0)^2\| (-\phi^\prime)^{1/2} FWF^{-1} u \|_{\mathcal{L}^2(M)}^2  \Big),
\end{gather}
where in the last inequality we used Equation \eqref{Eq: chi bound}.
As for the second term in Equation \eqref{Eq:Stress-Energy=Commutator}, using that $S_F\in x^2L^\infty(M)$, we establish the bound
\begin{equation*}
2 |Re \langle S_F u, V^\prime u \rangle| \leq \widetilde{C} \left(\| \phi^{1/2} \ u \|^2_{\mathcal{L}^2(M)} + \| \phi^{1/2} \ d_F u \|^2_{\mathcal{L}^2(M)}  \right),
\end{equation*}
for a suitable constant $\widetilde{C}>0$. Using Equation \eqref{Eq: chi bound} and the Poincar\'e inequality, this last bound becomes
\begin{equation}\label{Eq: Twisting Function Bound}
2 |Re \langle S_F u, V^\prime u \rangle| \leq C \delta^{-1}(\tau_1-\tau_0)^2 \|(-\phi^\prime)^{1/2} d_F u \|^2_{\mathcal{L}^2(M)}.
\end{equation}
At last we give a give a bound for the last term in Equation \eqref{Eq: Commutator P_Theta V^prime}, containing the pseudodifferential operator $\Theta$ which implements the 
boundary conditions. Recalling Hypothesis \ref{hypothesis_WF_Theta}, it is convenient to consider the following three cases separately 
\begin{itemize}
\item[a)]  $\Theta \in \Psi^k_b(\partial M)$ with $k \leq 1$,
\item[b)] $\Theta \in \Psi^k_b(\partial M)$ with $1<k\leq 2$.
\end{itemize}
Now we give a bound case by case.
\begin{itemize}
\item[a)] It suffices to focus on $\Theta \in \Psi^1_b(\partial M)$ recalling that, for $k<1$, $\Psi^k_b(\partial M) \subset \Psi^1_b(\partial M)$.
If with a slight abuse of notation we denote with $\Theta$ both the operator on the boundary and its trivial extension to the whole manifold, we can write 
$$\langle \Theta \gamma_- u, \gamma_- V^\prime u \rangle = \langle \widehat{N}(\Theta)(-i\nu_-) \gamma_- u, \gamma_ -V^\prime u \rangle =  \langle \gamma_-  \Theta u, \gamma_- V^\prime u \rangle,$$
where $\widehat{N}(\Theta)(-i\nu_-)$ is the indicial family as in Equation \eqref{Eq: indicial family}. 
We recall that any $A \in \Psi^s_b(\partial M)$, $s\in\mathbb{N}$, can be decomposed as $\sum\limits_{i=1}^n Q_i A_i + B$, with $A_i,B \in \Psi^{s-1}_b(\partial M)$, while $Q_i$, $i=1,\dots,n$ is a generating set of $\mathbf{Diff}^1_\nu(M)$. Hence we can rewrite $\Theta$ as
$$\Theta=\sum_i Q_i \Theta_i + \Theta^\prime = \sum_i \left( \Theta_i Q_i + [Q_i,\Theta_i] \right) + \Theta^\prime,$$
where $\Theta_i,\Theta^\prime$ and $[Q_i,\Theta_i]$ are in $\Psi^0(\partial M)$.
Therefore 
\begin{equation*}
|  \langle \gamma_-  \Theta u, \gamma_- V^\prime u \rangle | \leq | \langle \gamma_- \left( \sum_i \Theta_i Q_i  u \right), \gamma_- V^\prime u \rangle | +  | \langle \gamma_- \left( \left( [Q_i,\Theta_i] + \Theta^\prime \right)  u \right), \gamma_- V^\prime u \rangle |.
\end{equation*}
To begin with, we focus on the first term on the right hand side of this inequality. Using Equations \eqref{Prop: Bound_Gamma_Minus} and \eqref{Eq: chi bound} together with the Poincar\'e inequality \eqref{Eq: Poincare} and Lemma \ref{Lemma: PDO zero order bounded},
\begin{equation*}
\begin{split}
| \langle \gamma_- \left( \sum_i \Theta_i Q_i  u \right), \gamma_- V^\prime u \rangle | \leq \varepsilon \left( \sum_i \| 	\phi^{1/2} \Theta_i   Q_i u \|_{\mathcal{H}^1(M)}^2 + \|\phi^{1/2}  FWF^{-1} u\|_{\mathcal{H}^1(M)}^2 \right) + \\
+ C_\varepsilon \left( \sum_i \ \| \phi^{1/2}  \ Q_i u \|_{\mathcal{L}^2(M)}^2 + \| \phi^{1/2} FWF^{-1} u \|_{\mathcal{L}^2(M)}^2	\right) \leq C_\varepsilon \delta^{-1}(\tau_1-\tau_0)^2 \| (-\phi^\prime)^{1/2} d_F u \|^2_{\mathcal{L}^{2}(M)},
\end{split}
\end{equation*}
for a suitable constant $C_\varepsilon>0$.
As for the second term, since $u \in \mathcal{H}^{1,1}(M)$ we can proceed as above using that the operator $\Theta^\prime+[Q_i,\Theta_i]$ is of order $0$ and we can conclude that
\begin{equation*}
\left| \langle \gamma_- \left( \left( [Q_i,\Theta_i] + \Theta^\prime \right)  u \right), \gamma_- V^\prime u \rangle \right| \leq \widetilde{C}_\varepsilon \|  \phi^{1/2}  \ u \|_{\mathcal{H}^1(M)}^2 \leq C_\varepsilon \delta^{-1}(\tau_1-\tau_0)^2 \|  (-\phi^\prime)^{1/2} d_F u \|_{\mathcal{L}^{2}(M)}^2,
\end{equation*}
for suitable positive constants $C_\varepsilon$ and $\widetilde{C}_\varepsilon$.
Therefore, it holds a bound of the form
$$|Re \langle \Theta \gamma_- u, \gamma_- V^\prime u\rangle| \leq C^\prime_\epsilon \delta^{-1}(\tau_1-\tau_0)^2 \|  (-\phi^\prime)^{1/2} d_F u \|_{\mathcal{L}^{2}(M)}^2. $$

\item[b)] Since $\Psi^k_b(\partial M)\subset\Psi^{k^\prime}_b(\partial M)$ if $k<k^\prime$, it is enough to consider $\Theta\in\Psi^2_b(\partial M)$ and to observe that, we can decompose $\Theta$ as 
$$\Theta=\sum\limits_{i=1}^nQ_i\left(\sum\limits_{j=1}^n Q_j A_{ij}\right)+B_i,$$
where $B_i\in\Psi^1_b(\partial M)$ while $A_{ij}\in\Psi^0_b(\partial M)$. At this point one can apply twice consecutively  the same reasoning as in item a) to draw the sought conclusion. 
\end{itemize}

Finally, considering Equation \eqref{Eq:Stress-Energy=Commutator} and collecting all bounds we proved, we obtain
\begin{equation}
\langle -T_{\widehat{g}}^{ij}(W,\bigtriangledown_{\widehat{g}} \phi) Q_i u, Q_j u \rangle \leq C \Big( \| P_\Theta u \|_{\dot{\mathcal{H}}^{-1,1}(M)(}^2 + C \delta^{-1} (\tau_1-\tau_0)^2  \| (-\phi^\prime)^{1/2} d_F u \|_{\mathcal{L}^2(M)}^2.  
\end{equation}
Since the inner product $H$ defined by the left hand side of Equation \eqref{Eq: Dirichlet Form H} is positive definite, then for $\delta$ large enough
$$ \langle -T^{ij} _{\widehat{g}}(W,\bigtriangledown_{\widehat{g}} \phi)Q_i u, Q_j u \rangle - C \delta^{-1} (\tau_1-\tau_0)^2  \| (-\phi^\prime)^{1/2} d_F u \|_{\mathcal{L}^2(M)}^2  \geq 0,$$ 
and the associated Dirichlet form $\widetilde{\mathcal{Q}}$ defined as
\begin{equation}
\widetilde{\mathcal{Q}}(u,u) = \int_M  \left[ H( (-\phi^\prime)^{\/2} d_F u, (-\phi^\prime)^{1/2} d_F \overline{u}) - C \delta^{-1} (\tau_1-\tau_0)^2 |(-\phi^\prime)^{1/2} d_F u|^2  \right] x^2 d\mu_g,
\end{equation}
bounds $\| (-\phi^\prime)^{1/2} d_F u \|^2_{\mathcal{L}^2(M)}$.
We conclude the proof by observing that, once we have an estimate for $\| (-\phi^\prime)^{1/2} d_F u \|^2_{\mathcal{L}^2(M)}$,  with the Poincar\'e inequality we can also bound $\| (-\phi^\prime)^{1/2} u \|_{\mathcal{L}^2(M)}$. Therefore, considering the support of $\chi$ and $u$, there exists a compact subset $K \subset M$ such that
\begin{equation}
\| (-\phi^\prime)^{1/2} u \|_{\mathcal{L}^2(M)} \leq C \|(-\phi^\prime)^{1/2} P_\Theta u \|_{\dot{\mathcal{H}}^{-1,1}(K)},
\end{equation} 
from which the sought thesis descends.

\end{proof}

\begin{rem}
The case with $\Theta \in \Psi^k(M)$ of order $k \leq 0$, can also be seen as a corollary of the well-posedness result of \cite{GaWr18}.
\end{rem}

The following two statements guarantee uniqueness and existence of the solutions for the Klein-Gordon equation associated to the operator $P_\Theta$ individuated in Equation \eqref{Eq: PTheta}. Mutatis mutandis, since we assume that $\Theta$ is local in time, the proof of the first statement is identical to the counterpart in \cite{Vasy12} and therefore we omit it.

\begin{corollary}\label{Cor: uniqueness of the solution}
Let $M$ be a globally hyperbolic, asymptotically anti-de Sitter spacetime, {\it cf.} Definition \ref{Def: asymptotically AdS} and let $f \in \dot{\mathcal{H}}^{-1,1}(M)$ be vanishing whenever $\tau<\tau_0$, $\tau_0\in\mathbb{R}$. Suppose in addition that $\Theta$ abides to the Hypothesis \ref{hypothesis_WF_Theta}. Then there exists at most one $u \in \mathcal{H}^1_0(M)$ such that $supp(u) \subset \{ q \in M \ | \ \tau(q)\geq \tau_0 \}$ and it is a solution of $P_\Theta u=f$
\end{corollary}

\noindent At the same time the following statement holds true.

\begin{lemma}\label{Lem: existence of the solution}
Let $M$ be a globally hyperbolic, asymptotically anti-de Sitter spacetime, {\it cf.} Definition \ref{Def: asymptotically AdS} and let $f \in \dot{\mathcal{H}}^{-1,1}(M)$ be vanishing whenever $\tau<\tau_0$, $\tau_0\in\mathbb{R}$. Then there exists $u \in \mathcal{H}^{1,-1}(M)$ of the problem $P_\Theta u = f$, {\it cf.} Equation \eqref{Eq: PTheta}, such that $\tau(\textrm{supp}(u)) \geq \tau_0$.
\end{lemma}

The proof follows the one given in \cite[Prop. 4.15]{Vasy12}, but we feel worth sketching the main ideas. The first step consists of proving a local version of the lemma, namely that given a compact set $I \subset \mathbb{R}$, there exists $\sigma > 0$ such that for every $\tau_0 \in I$ there exists $u \in \mathcal{H}^{1,-1}(M)$ such that $supp(u)=\{p \in M \ | \ \tau(p) \geq 0\}$ and $P_\Theta u = f$ for $\tau < \tau_0+\sigma$. The main point of this part of the proof consists of applying Lemma \ref{Lemma: Bound u} to ensure that the adjoint of the Klein-Gordon operator, say $P^*_\Theta$, is invertible over the set of smooth functions supported in suitable compact subsets of $M$ -- see \cite[Lem. 4.14]{Vasy12} for further details. With this result in hand, one divides the time direction into sufficiently small intervals $[\tau_j,\tau_{j+1}]$ and uses a partition of unity along the time coordinate to build a global solution for $P_\Theta u = f$.

At last we extend our results for $u \in \mathcal{H}^{1,m}_{loc}(M)$ and for $f \in \dot{\mathcal{H}}^{-1,m+1}_{loc}(M)$. Let us consider $\Theta\in\Psi^k_b(\partial M)$ with $k\leq 0$, the proof for the positive cases being the same. If $m \geq 0$, Lemma \ref{Lemma: Bound u} entails that Equation \eqref{Eq: Strong boundary KG} admits a unique solution lying in $\mathcal{H}^1_{loc}(M)$. By the propagation of singularities theorem, {\it cf.} Theorem \ref{Thm: main theorem k negativo} and using Hypothesis \ref{hypothesis_WF_Theta}, the solution lies in $\mathcal{H}^{1,m}_{loc}(M)$ and the following generalization of the bound in Lemma \ref{Lemma: Bound u} holds true:
$$\| u \|_{\mathcal{H}^{1,m}(M)} \leq C \| f \|_{\dot{\mathcal{H}}^{-1,m+1}(M)}. $$
If $m<0$ we can draw the same conclusion considering, as in \cite[Thm. 8.12]{Vasy12}, 
\begin{equation}
P_\Theta u_j = f_j\\
\end{equation}
where $f_j \in \dot{\mathcal{H}}^{-1,m+1}(M)$ is sequence converging to $f$ as $j\to\infty$. Each of these equations has a unique solution $u_j\in\mathcal{H}^1(M)$. In addition the propagation of singularities theorem, {\it cf.} Theorem \eqref{Thm: main theorem k negativo} yields the bound
\begin{equation*}
\| u_k - u_j \|_{\mathcal{H}^{1,m}(K)} \leq C \| f_k - f_j \|_{\dot{\mathcal{H}}^{-1,m+1}(L)}
\end{equation*} 
for suitable compact sets $K,L \subset M$ and for every $j,k \in \mathbb{N}$. Since $f_j \rightarrow f$ in $\dot{\mathcal{H}}^{-1,m+1}(L)$, we can conclude that the sequence $u_j$ is converging to $u \in \mathcal{H}^{1,m}(K)$. Considering $f_j$ such that each $f_j$ vanishes if $\{\tau<\tau_0 \}$, one obtains the desired support property of the solution. To conclude this analysis we summarize the final result which combines Corollary \ref{Cor: uniqueness of the solution} and Lemma \ref{Lem: existence of the solution}.

\begin{prop}\label{Prop: existence and uniqueness}
Let $M$ be a globally hyperbolic, asymptotically anti-de Sitter spacetime, {\it cf.} Definition \ref{Def: asymptotically AdS} and let $m,\tau_0 \in \mathbb{R}$ while $f \in \dot{\mathcal{H}}^{-1,m+1}_{loc}(M)$. Assume in addition that $\Theta$ abides to Hypothesis \ref{hypothesis_WF_Theta}. If $f$ vanishes for $\tau<\tau_0$, $\tau_0\in\mathbb{R}$ being arbitrary but fixed, then there exists a unique $u \in \mathcal{H}^{1,m}_{loc}(M)$ such that
\begin{equation}
P_\Theta u = f,
\end{equation}
where $P_\Theta$ is the operator in Equation \eqref{Eq: PTheta}.
\end{prop} 

\noindent We have gathered all ingredients to prove the existence of advanced and retarded fundamental solutions associated to the Klein-Gordon operator $P_\Theta$, {\it cf.} Equation \eqref{Eq: PTheta}. To this end let us define the following notable subspaces of $\mathcal{H}^{k,m}(M)$, $k=0,\pm 1$, $m\in\mathbb{N}\cup\{0\}$: 

\begin{subequations}
\begin{equation}\label{Eq: suspsaces of H1(M) -}
	\mathcal{H}^{k,m}_-(M) = \{ u \in \mathcal{H}^{k,m}(M) \; | \;\exists\tau_-\in\mathbb{R}\;\textrm{such that}\; p\notin\textrm{supp}(u),\;\textrm{if}\, \tau(p)<\tau_-  \},
\end{equation}	
	\begin{equation}
		\label{Eq: suspsaces of H1(M) +}
	\mathcal{H}^{k,m}_+(M) = \{ u \in \mathcal{H}^{k,m}(M) \; | \;\exists\tau_+\in\mathbb{R}\;\textrm{such that}\; p\notin\textrm{supp}(u)\;\textrm{if}\, \tau(p)>\tau_+ \},
\end{equation}
\begin{equation}\label{Eq: timelikecompact H1(M)}
		\mathcal{H}^{k,m}_{tc}(M)\doteq \mathcal{H}^{k,m}_+(M)\cap\mathcal{H}^{k,m}_-(M),
\end{equation}
\end{subequations}	
where the subscript $tc$ stands for {\em timelike compact}. In addition we call
\begin{equation}\label{Eq: timelikecompact H1(M) with boundary condition}
	\mathcal{H}^{1,m}_{\pm,\Theta}(M)\doteq \{u\in H^{1,m}_\pm(M)\;|\; \gamma_+(u)=\Theta\gamma_-(u) \},
\end{equation}
where $\gamma_-,\gamma_+$ are the trace maps introduced in Theorem \ref{Thm: gamma-} and in Lemma \ref{Lem: gamma+}, while $\Theta$ is a pseudodifferential abiding to Hypothesis \ref{hypothesis_WF_Theta}.

Exactly as in \cite{GaWr18} from Lemma \ref{Lemma: Bound u} and from Proposition \ref{Prop: existence and uniqueness}, it descends the following result on the advanced and retarded propagators $G_{\Theta}^\pm$ associated to the Klein-Gordon operator $P_\Theta$, {\it cf.} Equation \eqref{Eq: PTheta}. 

\begin{thm}\label{Thm: Existence_Uniqueness_Propagators}
Let $P_\Theta$ be the Klein-Gordon operator as per Equation \eqref{Eq: PTheta} where $\Theta$ abides to Hypothesis \ref{hypothesis_WF_Theta}. Then there exist unique retarded $(+)$ and advanced $(-)$ propagators, that is continuous operators $G_{\Theta}^\pm:\dot{\mathcal{H}}^{-1,m+1}_\pm(M) \rightarrow \mathcal{H}^{1,m}_\pm(M)$ such that $P_\Theta G_{\Theta}^\pm = \mathbb{I}$ on $\dot{\mathcal{H}}^{-1,m+1}_\pm(M)$ and $G_{\Theta}^\pm P_\Theta = \mathbb{I}$ on $\mathcal{H}^{1,m}_{\pm,\Theta}(M)$. Furthermore, $G_{\Theta}^\pm$ is a continuous map from $\dot{\mathcal{H}}_0^{-1,\infty}(M)$ to $\mathcal{H}_{loc}^{1,\infty}(M)$ where the subscript $0$ indicates that we consider only functions of compact support. 
\end{thm}

Observe that the restriction to $\mathcal{H}^{1,m}_{\pm,\Theta}(M)$ is necessary since, per construction an element in the range of $G^\pm_\Theta P_\Theta$ abides to the boundary conditions as in Equation \eqref{Eq: Strong boundary KG}. 

\begin{rem}\label{Rem: causal propagator}
	Associated to the advanced and to retarded propagators, one can define the {\em causal propagator} $G_\Theta : \dot{\mathcal{H}}_{0}^{-1,m+1}(M) \rightarrow \mathcal{H}^{1,m}_{loc}(M)$ as $G_\Theta = G_{\Theta}^+-G_{\Theta}^-$. 
\end{rem}

Since $G^\pm_\Theta$ are continuous maps, {\it cf.} Theorem \ref{Thm: Existence_Uniqueness_Propagators}, one can apply Schwartz kernel theorem to infer that one can associate to them a bi-distribution $\mathcal{G}^\pm_\Theta\in\mathcal{D}^\prime(M\times M)$. To conclude the section we highlight a standard and important application of the fundamental solutions and in particular of the causal propagator {\it cf.} Remark \ref{Rem: causal propagator}. 

\begin{prop}\label{Prop: exact sequence of solutions}
	Let $P_\Theta$ be the Klein-Gordon operator as per Equation \eqref{Eq: PTheta} and let $G_\Theta$ be its associated causal propagator, {\it cf.} Remark \ref{Rem: causal propagator}. Then the following is an exact sequence:
	\begin{align}\label{Eq: short exact sequence for spacetimes with timelike boundary}
		0\to \mathcal{H}^{1,\infty}_{tc,\Theta}(M)\stackrel{P_\Theta}{\longrightarrow}\dot{\mathcal{H}}^{-1,\infty}_{tc}(M) 
		\stackrel{G_\Theta}{\longrightarrow}
		\mathcal{H}^{1,\infty}_\Theta(M)
		\stackrel{P_\Theta}{\longrightarrow}
		\dot{\mathcal{H}}^{-1,\infty}(M)\to 0\,.
	\end{align}
\end{prop}

\begin{proof}
	To prove that the sequence is exact, we start by establishing that $P_\Theta$ is injective on $\mathcal{H}^{1,\infty}_{tc,\Theta}(M)$. This is a consequence of Theorem \ref{Thm: Existence_Uniqueness_Propagators} which guarantees that, if $P_\Theta(h)=0$ for $h\in\mathcal{H}^{1,\infty}_{tc,\Theta}(M)$, then $G^+P_\Theta(h)=h=0$.
	
	Secondly, on account of Theorem \ref{Thm: Existence_Uniqueness_Propagators} and in particular of the identity $G^\pm_\Theta P_\Theta  =\mathbb{I}$ on $\mathcal{H}^1_{\pm,\Theta}(M)$, it holds that $G_\Theta P_\Theta(f)=0$ for all $f\in \mathcal{H}^{1,\infty}_{tc,\Theta}(M)$. Hence $\mathrm{Im}(P_\Theta)\subseteq\ker(P_\Theta)$. Assume that there exists $f\in\dot{\mathcal{H}}^{-1,\infty}_{tc}(M)$ such that $G_\Theta(f)=0$. It descends that $G^+_\Theta(f)=G^-_\Theta(f)\in\mathcal{H}^{1,\infty}_{tc,\Theta}(M)$. Applying $P_\Theta$ it holds that $f=P_\Theta G^+_\Theta(f)$, that is $f\in P_\Theta[\mathcal{H}^{1,\infty}_{tc,\Theta}(M)]$.

	The third step consists of recalling that, per construction, $P_\Theta G_\Theta =0$ and that, still on account of Theorem \ref{Thm: Existence_Uniqueness_Propagators}, $\textrm{Im}(G_\Theta)\subseteq\ker(P_\Theta)$. To prove the opposite inclusion, suppose that $u\in\ker(P_\Theta)$. Let $\chi\equiv\chi(\tau)$ be a smooth function such that there exists $\tau_0,\tau_1\in\mathbb{R}$ such that $\chi=1$ if $\tau>\tau_1$ and $\chi=0$ if $\tau<\tau_0$. Since $\Theta$ is a static boundary condition and, therefore, it commutes with $\chi$, it holds that $\chi u\in\mathcal{H}^{1,\infty}_{+,\Theta}(M)$. Hence setting $f\doteq P_\Theta\chi u$, a direct calculation shows that $G_\Theta f=u$
	
	To conclude we need to show that the map $P_\Theta$ on the before last arrow is surjective. To this end, let $j\in\dot{\mathcal{H}}^{-1,\infty}(M)$ and let $\chi\equiv\chi(\tau)$ be as above. Let $h\doteq G^+_\Theta\left(\chi j\right)+G^-_\Theta\left((1-\chi)j\right)$. Per construction $h\in\mathcal{H}^{1,\infty}(M)$ and $P_\Theta(h)=j$.
\end{proof}

\noindent Mainly for physical reasons we individuate the following special classes of boundary conditions. Recall that, according to Theorem \ref{Thm: globally hyperbolic} $M$ is isometric to $\mathbb{R}\times\Sigma$ and $\partial M$ to $\mathbb{R}\times\partial\Sigma$. 

\begin{defn}\label{Def: admissible boundary conditions}
	Let $\Theta\in\Psi^k_b(M)$ with $k\leq 2$ and let $\Theta=\Theta^*$ We call $\Theta$
	\begin{itemize}
		\item {\em physically admissible} if $WF_b^{-1,s+1}(\Theta u)\subseteq WF_b^{-1,s+1}(P_\Theta u)$ for all $u\in\mathcal{H}^{1,m}_{loc}(M)$ with $m\leq 0$ and $s\in\mathbb{R}\cup\{\infty\}$. 
		\item a {\em static} boundary condition if $\Theta\equiv\Theta_K$ is the natural extension to $\Psi^k_b(M)$ of a pseudodifferential operator $K=K^*\in\Psi_b^k(\partial\Sigma)$ with $k\leq 2$.
	\end{itemize}
\end{defn}

Observe that any static boundary condition is automatically local in time, see Definition \ref{Def: pseudo local in time}. Starting from these premises we can investigate further properties of the fundamental solutions, starting from the singularities of the advanced and retarded propagators. To this end let us introduce  $\mathcal{W}_b^{-\infty}(M)$  the space of bounded operators from $\dot{\mathcal{H}}_0^{-1,-\infty}(M)$ to $\mathcal{H}^{1,\infty}_{loc}(M)$ and we give a definition of wavefront set complementary to that of Definition \ref{Def: WF of PsiDO}.

\begin{defn}[Operatorial wavefront set $WF_b^{Op}(M)$]\label{Def: WF operatorial}
	Let $\Lambda : \dot{\mathcal{H}}^{-1,-\infty}_0(M) \rightarrow \mathcal{H}^{1,\infty}_{loc}(M)$ be a continuous map. A point $(q_1,q_2) \in {}^bS^*M \times {}^b S^* (M) \not \in WF_b^{Op}(M) $ if there exists two b-pseudodifferential operators $B_1$ and $B_2$ in $\Psi_b^0(M)$ elliptic at $q_1$ and $q_2$ respectively, such that $B_1 \Lambda B_2^* \in \mathcal{W}_b^{-\infty}(M)$. 
\end{defn}

Recalling Equation \eqref{Eq: cosphere bundle}, we can state the following theorem characterizing the singularities of the advanced and of the retarded fundamental solutions. The proof is a direct application of Theorem \ref{Thm: main theorem k positivo} or of Theorem \ref{Thm: main theorem k negativo}.

\begin{thm}\label{thm:wavefront_set_propagator}
	Let $\Delta$ denote the diagonal in ${}^b S^*M \times {}^b S^*M$ and let $\Theta$ be physically admissible as per Definition \ref{Def: admissible boundary conditions}. Then
	$$WF_b^{Op}(G_{\Theta}^\pm)  \setminus \Delta \subset \{ (q_1,q_2) \in {}^b S^*M \times {}^b S^*M \ | \ q_1 \dot{\sim} q_2, \ \pm t(q_1)>\pm t(q_2) \},$$
	where $q_1 \dot{\sim} q_2$ means that $q_1,q_2$ are two points in $\dot{\mathcal{N}}$, {\it cf.} Equation \eqref{Eq: compressed characteristic set} connected by a generalized broken bicharacteristic, {\it cf.} Definition \ref{Def: generalized broken bicharacteristics}.
\end{thm}

\begin{rem}
	The reason for the hypothesis on $\Theta$ lies in the fact that we do not want to alter the microlocal behavior of the system in $\mathring{M}$. More precisely, if no restriction on the wavefront set of $\Theta u$ is placed, then by the propagation of singularities theorem, {\it cf.} Theorem \ref{Thm: main theorem k positivo}, in addition to the singularities propagating along the generalized broken bicharacteristics of the Klein-Gordon operator we should account also for those of $\Theta u$. On the one hand this would be in sharp contrast with what happens if $M$ were a globally hyperbolic spacetime without boundary. On the other hand, in concrete applications such as the construction of Hadamard two-point functions, one seeks for bi-distributions with a prescribed form of the wave front set and whose antisymmetric part coincides with the difference between the advanced and retarded fundamental solutions associated to the Klein-Gordon operator with boundary condition implemented by $\Theta$, see {\it e.g.} \cite{DF16,DF18,Dybalski:2018egv,Wro17,GaWr18}.
\end{rem}

\noindent In addition one can infer the following localization property which is sometimes referred to as {\em time-slice axiom}.

\begin{corollary}\label{Cor: time_slice_axiom}
	Let $\mathcal{H}^{-1,\infty}_{tc,[\tau_1,\tau_2]}(M)\subset\mathcal{H}^{-1,\infty}_{tc}(M)$ be the collection of all $u\in\mathcal{H}^{-1,\infty}_{tc}(M)$ such that $p\notin\textrm{supp}(u)$ whenever $\tau(p)\notin[\tau_1,\tau_2]$, $\tau_1,\tau_2\in\mathbb{R}$. Then, if $\Theta$ is a static boundary condition as per Definition \ref{Def: admissible boundary conditions}, the inclusion map $\iota_{\tau_1,\tau_2}: \dot{\mathcal{H}}^{-1,\infty}_{tc,[\tau_1,\tau_2]}(M) \rightarrow \dot{\mathcal{H}}^{-1,\infty}_{tc}(M)$ induces the isomorphism
	\begin{equation}\label{Eq: isomorphism initial data}
		[\iota_{\tau_1,\tau_2}] : \dfrac{\dot{\mathcal{H}}_{tc,[\tau_1,\tau_2]}^{1,\infty}(M)}{P_\Theta \mathcal{H}_{tc,[\tau_1,\tau_2]}^{1,\infty}(M)} \rightarrow \dfrac{\dot{\mathcal{H}}_{tc}^{1,\infty}(M)}{P_\Theta \mathcal{H}_{tc}^{1,\infty}(M)}.
	\end{equation}
\end{corollary}

\begin{proof}
By direct inspection one can realize that the map $\iota_{\tau_1,\tau_2}$ descends to the quotient space $\dfrac{\dot{\mathcal{H}}_{tc,[\tau_1,\tau_2]}^{1,\infty}(M)}{P_\Theta \mathcal{H}_{tc,[\tau_1,\tau_2]}^{1,\infty}(M)}$. The ensuing application $[\iota_{\tau_1,\tau_2}]$ is manifestly injective. We need to show that it is also surjective. Consider therefore any $[f]\in\dfrac{\dot{\mathcal{H}}_{tc}^{1,\infty}(M)}{P_\Theta \mathcal{H}_{tc}^{1,\infty}(M)}$ and let $G_{\Theta}(f)$ be the associated solution of the Klein-Gordon equation, {\it cf.} Equation \eqref{Eq: short exact sequence for spacetimes with timelike boundary}. Let $\chi\equiv\chi(\tau)$ be a smooth function such that $\chi=1$ if $\tau>\tau_2$ while $\chi=0$ if $\tau<\tau_1$. The function $h\doteq P_\Theta\left(\chi G_\Theta(f)\right)\in\dot{\mathcal{H}}^{-1,\infty}_{tc,[\tau_1,\tau_2]}(M)$, where $G_\Theta$ is the causal propagator, {\it cf.} Remark \ref{Rem: causal propagator} and Proposition \ref{Prop: exact sequence of solutions}. Per construction the map $P_\Theta\circ\chi\circ G_\Theta$ descends to an application from $\dfrac{\dot{\mathcal{H}}_{tc}^{1,\infty}(M)}{P_\Theta \mathcal{H}_{tc}^{1,\infty}(M)}$ to $\dfrac{\dot{\mathcal{H}}_{tc,[\tau_1,\tau_2]}^{1,\infty}(M)}{P_\Theta \mathcal{H}_{tc,[\tau_1,\tau_2]}^{1,\infty}(M)}$ which is both a left and a right inverse of $[\iota_{\tau_1,\tau_2}]$.
\end{proof}

\section{Hadamard States}\label{Sec: Hadamard two-point distributions}

In this section, we discuss a specific application of the results obtained in the previous section, namely we prove existence of a family of distinguished two-point correlation functions for a Klein-Gordon field on a globally hyperbolic, asymptotically AdS spacetime, dubbed {\em Hadamard two-point distributions}. These play an important r\^{o}le in the algebraic formulation of quantum field theory, particularly when the underlying background is a generic globally hyperbolic spacetime with or without boundary, see {\it e.g.} \cite{Khavkine:2014mta} for a review as well as \cite{DF16,DF18,DFM18} for the analysis on anti-de Sitter spacetime and \cite{Wro17} for an that on a generic asymptotically AdS spacetime, though only in the case of Dirichlet boundary conditions.  

Here our goal is to prove that such class of two-point functions exists even if one considers more generic boundary conditions. To prove this statement, the strategy that we follow is divided in three main steps, which we summarize for the reader's convenience. To start with, we restrict our attention to static, asymptotically anti-de Sitter and globally hyperbolic spacetimes and to boundary conditions which are both physically acceptable and static, see Definition \ref{Def: admissible boundary conditions}. In this context, by means of spectral techniques, we give an explicit characterization of the advanced and retarded fundamental solutions. To this end we use the theory of boundary triples, a framework which is slightly different, albeit connected, to the one employed in the previous sections, see \cite{DDF18}. 

Subsequently we show that, starting from the fundamental solutions and from the associated causal propagator, it is possible to identify a distinguished two-point distributions of Hadamard form.

To conclude, we adapt and we generalize to the case in hand a deformation argument due to Fulling, Narcowich and Wald, \cite{FNW81} which, in combination with the propagation of singularities theorem, allows to infer the existence of Hadamard two-point distributions for a Klein-Gordon field on a generic globally hyperbolic and asymptotically AdS spacetime starting from those on a static background.

\subsection{Fundamental solutions on static spacetimes}\label{Sec: Static Fundamental Solutions}

In this section we give a concrete example of advanced and retarded fundamental solutions for the Klein-Gordon operator $P_\Theta$, {\it cf.} Equation \eqref{Eq: PTheta} on a static, globally hyperbolic, asymptotically AdS spacetime. For the sake of simplicity, we consider a massless scalar field, corresponding to the case $\nu = (n-1)/2$, see Equation \ref{Eq: nu parameter}. Observe that, since the detailed analysis of this problem has been mostly carried out in \cite{DDF18}, we refer to it for the derivation and for most of the technical details. Here we shall limit ourselves to giving a succinct account of the main results.

As a starting point, we specify precisely the underlying geometric structure:
\begin{defn}\label{Def: standard static}
	Let $(M,g)$ be an $n$-dimensional Lorentzian manifold. We call it a static globally hyperbolic, asymptotically AdS spacetime if it abides to Definition \ref{Def: asymptotically AdS} and, in addition,
	\begin{itemize}
		\item[1)] There exists an irrotational, timelike Killing field $\chi\in\Gamma(TM)$, such that $\mathcal{L}_\chi(x)=0$ where $x$ is the global boundary function,

		\item[2)] $(M,\hat{g})$ is isometric to a standard static spacetime, that is a warped product $\mathbb{R}\times_\beta S$ with line element $ds^2=-\alpha^2 dt^2+ h_S$ where $h_S$ is a $t$-independent Riemannian metric on $S$, while $\alpha\neq\alpha(t)$ is a smooth, positive function. 
	\end{itemize}
\end{defn}

\begin{rem}\label{Rem: our standard static}
	In the following, without loss of generality, we shall assume that, whenever we consider a static globally hyperbolic, asymptotically flat spacetime if it abides to Definition \ref{Def: asymptotically AdS}, the timelike Killing field $\chi$ coincides with the vector field $\partial_\tau$, {\it cf.} Theorem \ref{Thm: globally hyperbolic}. Hence the underlying line-element reads as $ds^2=-\beta d\tau^2 +\kappa$ where both $\beta$ and $\kappa$ are $\tau$-independent and $S$ can be identified with the Cauchy surface $\Sigma$ in Theorem \ref{Thm: globally hyperbolic}. For convenience we also remark that, in view of this characterization of the metric, the associated Klein-Gordon equation $Pu=0$ with $P=\Box_g$ reads
	\begin{equation}\label{Eq: static KG operator}
		\left(-\partial^2_\tau+ E\right)u=0,
	\end{equation}
	where $E=\beta\Delta_\kappa$, being $\Delta_\kappa$ the Laplace-Beltrami operator associated to the the Riemannian metric $\kappa$.
	
\end{rem}

Henceforth we consider only {\em static} boundary conditions as per Definition \ref{Def: admissible boundary conditions} which we indicate with the symbol $\Theta_K$ to recall that they are induced from $K\in\Psi^k_b(\partial M)$. Since the underlying spacetime is static, in order to construct the advanced and retarded fundamental solutions, we can focus our attention on $\mathcal{G}_{\Theta_K}\in \mathcal{D}^\prime(\mathring{M} \times \mathring{M})$ , the bi-distribution associated to the causal propagator $G_{\Theta_K}$, {\it cf.} Remark \ref{Rem: causal propagator}. It satisfies the following initial value problem, see also \cite{DDF18}:

\begin{equation}\label{eq: equations fundamental solutions}
	\begin{cases}
		(P_{\Theta_K} \otimes \mathbb{I}) \mathcal{G}_{\Theta_K} = (\mathbb{I} \otimes P_{\Theta_K})\mathcal{G}_{\Theta_K} = 0 \\
		\mathcal{G}_{\Theta_K}|_{\tau=\tau^\prime} = 0 \quad\\
		\partial_\tau\mathcal{G}_{\Theta_K}|_{\tau=\tau^\prime}= - \partial_{\tau^\prime} \mathcal{G}_{\Theta_K}|_{\tau=\tau^\prime} = \delta
	\end{cases}
\end{equation}
where $\delta$ is the Dirac distribution on the diagonal of $\mathring{M} \times \mathring{M}$. Starting from $\mathcal{G}_{\Theta_K}$ one can recover the advanced and retarded fundamental solutions $\mathcal{G}^\pm_{\Theta_K}$ via the identities:
\begin{equation}\label{eq: adv/ret from causal}
	\mathcal{G}^-_{\Theta_K}=\vartheta(\tau-\tau^\prime)\mathcal{G}_{\Theta_K}\quad\textrm{and}\quad\mathcal{G}^+_{\Theta_K}=-\vartheta(\tau^\prime-\tau)\mathcal{G}_{\Theta_K},
\end{equation}
where $\vartheta$ is the Heaviside function. The existence and the properties of $\mathcal{G}_{\Theta_K}$ have been thoroughly analyzed in \cite{DDF18} using the framework of boundary triples, {\it cf.} \cite{Grubb68}. Here we recall the main structural aspects.

\begin{defn}\label{Def: boundary triples}
	Let $H$ be a separable Hilbert space over $\mathbb{C}$ and let $S:D(S)\subset H \rightarrow H$ be a closed, linear and symmetric operator.
	A {\em boundary triple} for the adjoint operator $S^*$ is a triple $(\mathsf{h},\gamma_0,\gamma_1)$, where $\mathsf{h}$ is a separable Hilbert space over $\mathbb{C}$ while $\gamma_0,\gamma_1:D(S^*) \rightarrow \mathsf{h}$ are two linear maps satisfying
	\begin{itemize}
		\item[1)] For every $f,f^\prime \in D(P^*)$ it holds
		\begin{equation}\label{eq:LagrangeId}
			(S^*f|f^\prime)_H-(f|S^*f^\prime)_H = (\gamma_1 f | \gamma_0 f^\prime)_{\mathsf{h}} - (\gamma_0 f | \gamma_1 f^\prime)_{\mathsf{h}}
		\end{equation}
		\item[2)] The map $\gamma:D(S^*) \rightarrow\mathsf{h}\times\mathsf{h}$ defined by $\gamma(f) = (\gamma_0 f, \gamma_1 f)$ is surjective.
	\end{itemize}
\end{defn}

\noindent One of the key advantages of this framework is encoded in the following proposition, see \cite{Mal92} 

\begin{prop}\label{Prop: self-adjoint extensions via boundary triples}
	Let $S$ be a linear, closed and symmetric operator on $H$. Then an associated boundary triple $(\mathsf{h},\gamma_0,\gamma_1)$ exists if and only if $S^*$ has equal deficiency indices. In addition, if $\Theta:D(\Theta) \subseteq\mathsf{h}\rightarrow\mathsf{h}$ is a closed and densely defined linear operator, then $S_\Theta \doteq S^*|_{ker(\gamma_1-\Theta \gamma_0)}$ is a closed extension of $S$ with domain
	\begin{equation*}
		D(S_\Theta) \doteq \{ f \in D(S^*)\; |\; \gamma_0(f) \in D(\Theta), \; \textrm{and}\; \gamma_1(f) = \Theta\gamma_0(f) \}
	\end{equation*}
	The map $\Theta \mapsto S_\Theta$ is one-to-one and $S^*_\Theta = S_{\Theta^*}$. In other word there is a one-to-one correspondence between self-adjoint operators $\Theta$ on $\mathsf{h}$ and self-adjoint extensions of $S$.
\end{prop}


Noteworthy is the application of this framework to the case where the r\^{o}le of $S$ is played by a second order elliptic partial differential operator $E$. Observe that this symbol is employed having in mind the subsequent application to Equation \eqref{Eq: static KG operator}. To construct a boundary triple associated with $E^*$, let $n$ be the unit, outward pointing, normal of $\partial\Sigma$ and let
$$\Gamma_0\colon H^2(\Sigma)\ni f\mapsto\Gamma f\in H^{3/2}(\Sigma),\,\qquad\Gamma_1\colon H^2(\Sigma)\ni f\mapsto-\Gamma \nabla_n f\in H^{1/2}(\Sigma)\,,$$
where $H^k(\Sigma)$ indicates the Sobolev space associated to the Riemannian manifold $(\Sigma,\kappa)$ introduced at the end of Section \ref{Sec: Manifolds of bounded geometry}. Here $\Gamma:H^s(\Sigma)\to H^{s-\frac{1}{2}}(\Sigma)$, $s>\frac{1}{2}$ is the continuous surjective extension of the restriction map from $C^\infty_0(\Sigma)$ to $C^\infty_0(\partial\Sigma)$, {\it cf.} \cite[Th. 4.10 \& Cor. 4.12]{GS13}.  In addition, since the inner product $(\,|\,)_{L^2(\partial\Sigma)}$ on $L^2(\partial\Sigma)\equiv L^2(\partial\Sigma;\iota^*_\Sigma d\mu_g)$, $\iota_\Sigma:\partial\Sigma\hookrightarrow\Sigma$, extends continuously to a pairing on $H^{-1/2}(\partial\Sigma)\times H^{1/2}(\partial\Sigma)$ as well as on $H^{-3/2}(\partial\Sigma)\times H^{3/2}(\partial\Sigma)$, there exist isomorphisms
$$\iota_\pm\colon H^{\pm 1/2}(\partial\Sigma)\to L^2(\partial\Sigma),\qquad j_\pm\colon H^{\pm 3/2}(\partial\Sigma)\to L^2(\partial\Sigma)\,,$$
such that, for all $(\psi,\phi)\in H^{1/2}(\partial\Sigma)\times H^{-1/2}(\partial\Sigma)$ and for all $(\widetilde{\psi},\widetilde{\phi})\in H^{3/2}(\partial\Sigma)\times H^{-3/2}(\partial\Sigma)$, 
\begin{align*}
	(\psi,\phi)_{(1/2,-1/2)}=(\iota_+\psi|\,\iota_-\phi)_{L^2(\partial\Sigma)}\,,\quad
	(\widetilde{\psi},\widetilde{\phi})_{(3/2,-3/2)}=(j_+\widetilde{\psi}|\,j_-\widetilde{\phi})_{L^2(\partial\Sigma)}\,,
\end{align*}
where $(,)_{(1/2,-1/2)}$ and $(,)_{(3/2,-3/2)}$ stand for the duality pairings between the associated Sobolev spaces.

\begin{rem}
Note that in the massless case, the two trace operators $\Gamma_0$ and $\Gamma_1$ coincide respectively with the restriction to $H^2(M)$ of the traces $\gamma_-$ and $\gamma_+$ introduced in Theorem \ref{Thm: gamma-} and in Lemma \ref{Lem: gamma+}. 
\end{rem}

Gathering all the above ingredients, we can state the following proposition, {\it cf.} \cite[Thm. 24 \& Rmk 25]{DDF18}:

\begin{prop}\label{Proposition: boundary triple of the Laplacian}
	Let $E^*$ be the adjoint of a second order, elliptic, partial differential operator on a Riemannian manifold $(\Sigma,\kappa)$ with boundary and of bounded geometry. Let 
	\begin{gather}\label{Equation: Dirichlet boundary map for Laplacian}
		\gamma_0\colon H^2(M)\ni f\mapsto \iota_+\Gamma_0f\in L^2(\partial  M)\,,\\
		\label{Equation: Neumann boundary map for Laplacian}
		\gamma_1\colon H^2(M)\ni f\mapsto j_+\Gamma_1f\in L^2(\partial  M)\,,
	\end{gather}
	Then $(L^2(\partial  M),\gamma_0,\gamma_1)$ is a boundary triple for $E^*$.
\end{prop}

\noindent Combining all data together, particularly Proposition \ref{Prop: self-adjoint extensions via boundary triples} and Proposition \ref{Proposition: boundary triple of the Laplacian} we can state the following theorem, whose proof can be found in \cite[Thm 29]{DDF18}

\begin{thm}\label{Theorem: construction of the advanced and retarded propagators}
	Let $(M,g)$ be a static, globally hyperbolic, asymptotically AdS spacetime as per Definition \ref{Def: standard static}.
	Let $(\gamma_0,\gamma_1,L^2(\partial M))$ be the boundary triple as in Proposition \ref{Proposition: boundary triple of the Laplacian} associated with $E^*$, the adjoint of the elliptic operator defined in \eqref{Eq: static KG operator} and let $K$ be a densely defined self-adjoint operator on $L^2(\partial\Sigma)$ which individuates a static and physically admissible boundary condition as per Definition \ref{Def: admissible boundary conditions}.
	Let $E_K$ be the self-adjoint extension of $E$ defined as per Proposition \ref{Prop: self-adjoint extensions via boundary triples} by $E_K\doteq E^*|_{D(E_K)}$, where $D(E_K)\doteq\ker(\gamma_1-K\gamma_0)$.
	Furthermore, let assume that the spectrum of $E_{K}$ is bounded from below.\\
	Then, calling $\Theta_K$ the associated boundary condition, the advanced and retarded Green's operators $\mathsf{G}^\pm_{\Theta_K}$ associated to the wave operator
	$\partial_t^2+E_K$ exist and they are unique.
	They are completely determined in terms of $\mathcal{G}^\pm_{\Theta_K}\in\mathcal{D}^\prime(\mathring{M}\times\mathring{M})$.
	These are bidistributions such that $\mathcal{G}^-_{\Theta_K}=\vartheta(t-t^\prime)\mathcal{G}_{\Theta_K}$ and $\mathcal{G}^+_{\Theta_K}=-\vartheta(t^\prime-t)\mathcal{G}_{\Theta_K}$ where $\mathcal{G}_{\Theta_K}\in\mathcal{D}^\prime(\mathring{M}\times\mathring{M})$ is such that, for all $f\in \mathcal{D}(\mathring{M})$
	\begin{align}\label{Equation: construction of the causal propagator}
		\mathcal{G}_{\Theta_K}(f_1,f_2)\doteq\int_{\mathbb{R}^2}\textrm{d}t\textrm{d}t'\,\bigg( f_1(t)\bigg| (-E_K)^{-\frac{1}{2}}\sin\big[(-E_K)^{\frac{1}{2}}(t-t^\prime)\big]f_2(t^\prime)\bigg),
	\end{align}
	where $f(t)\in H^2(\Sigma)$ denotes the evaluation of $f$, regarded as an element of $C_{\textrm{c}}^\infty(\mathbb{R},H^\infty(\Sigma))$ and $E_K^{-\frac{1}{2}}\sin\big[E_K^{\frac{1}{2}}(t-t^\prime)]$ is defined exploiting the functional calculus for $E_K$.
	Moreover it holds that 
	$$\mathsf{G}^\pm_{\Theta_K}\colon\mathcal{D}(\mathring{M})\to C^\infty(\mathbb{R},H^\infty_{\Theta_K}(\Sigma))\,,$$
	where $H^\infty_{\Theta_K}(\Sigma)\doteq\bigcap_{k\geq 0}D(E_{\Theta_K}^k)$.
	In particular,
	\begin{align}\label{Equation: boundary condition for advanced and retarded propagators}
		\gamma_1\big(\mathsf{G}^\pm_{\Theta_K} f\big)=\Theta_K\gamma_0\big(\mathsf{G}^\pm_{\Theta_K} f\big)
		\qquad\forall f\in C^\infty_0(\mathring{M})\,.
	\end{align}
\end{thm}

\begin{rem}
	Observe that, in Theorem \ref{Theorem: construction of the advanced and retarded propagators} we have constructed the advanced and retarded fundamental solutions $\mathcal{G}^\pm_\Theta$ as elements of $\mathcal{D}^\prime(\mathring{M}\times\mathring{M})$. Yet we can combine this result with Theorem \ref{Thm: Existence_Uniqueness_Propagators} to conclude that there must exist unique and advanced retarded propagators on the whole $M$ whose restriction to $\mathring{M}$ coincides with $\mathcal{G}^\pm_{\Theta_K}$. With a slight abuse of notation we shall refer to these extended fundamental solutions with the same symbol.
\end{rem}

\subsection{Existence of Hadamard States on Static Spacetimes}\label{Sec: Existence of Hadamard States on Static Spacetimes}

In this section, we discuss the existence of Hadamard two-point functions. We stress that the so-called Hadamard condition and its connection to microlocal analysis have been first studied and formulated under the assumption that the underlying spacetime is without boundary and globally hyperbolic. We shall not enter into the details and we refer an interested reader to the survey in \cite{Khavkine:2014mta}.

As outlined in the introduction, if the underlying background possesses a timelike boundary, the notion of Hadamard two-point function needs to be modified accordingly. Here we follow the same rationale advocated in \cite{DF16, Dappiaggi:2017wvj} and also in \cite{Dybalski:2018egv, Wro17}.

\begin{defn}\label{Def: Hadamard 2-pt function}
Let $(M,g)$ be a globally hyperbolic, asymptotically AdS spacetime as per Definition \ref{Def: asymptotically AdS}. A bi-distribution $\lambda_2\in\mathcal{D}^\prime(M\times M)$ is called of {\em Hadamard form} if its restriction to $\mathring{M}$ has the following wavefront set
\begin{equation}\label{eq: Hadamard_Wavefront_Set}
WF( \lambda_2 ) = \left\{ (p,k,p^\prime,-k^\prime) \in T^*(\mathring{M}\times\mathring{M}) \setminus \{ 0 \}\; |\; (p,k) \sim (p^\prime,k^\prime)\;\textrm{and}\; k \triangleright 0  \right\},
\end{equation}
where $\sim$ entails that $(p,k)$ and $(p^\prime,k^\prime)$ are connected by a generalized broken bicharactersitic, while $k\triangleright 0$ means that the co-vector $k$ at $p\in\mathring{M}$ is future-pointing. Furthermore we call $\lambda_{2,\Theta}\in\mathcal{D}^\prime(M\times M)$ a {\em Hadamard two-point function} associated to $P_\Theta$, if, in addition to Equation \eqref{eq: Hadamard_Wavefront_Set}, it satisfies
$$(P_\Theta\otimes\mathbb{I})\lambda_{2,\Theta}=(\mathbb{I}\otimes P_\Theta)\lambda_{2,\Theta}=0,$$
and, for all $f,f^\prime\in\mathcal{D}(\mathring{M})$,
\begin{equation}\label{Eq: constraints on 2-pt function}
\lambda_{2,\Theta}(f,f)\geq 0,\quad\textrm{and}\quad\lambda_{2,\Theta}(f,f^\prime)-\lambda_{2,\Theta}(f^\prime,f)=i\mathcal{G}_\Theta(f,f^\prime),
\end{equation}
where $P_\Theta$ is the Klein-Gordon operator as in Equation \eqref{Eq: PTheta}, while $\mathcal{G}_\Theta$ is the associated causal propagator, {\it cf.} Remark \ref{Rem: causal propagator}.
\end{defn}

\begin{rem}\label{Rem: bulk-to-bulk two-point function}
	To make contact with the terminology often used in theoretical physics, given a Hadamard two-point function $\lambda_{2, \Theta}$, we can identify the following associated bidistributions:
	\begin{itemize}
	\item the {\em bulk-to-bulk two-point function } $\mathring{\lambda}_{2,\Theta}\in\mathcal{D}^\prime(\mathring{M}\times\mathring{M})$ such that $\mathring{\lambda}_{2,\Theta}\doteq\left.\lambda_{2,\Theta}\right|_{\mathring{M}}$ is the restriction of the Hadamard two-point function to $\mathring{M}\times\mathring{M}$. 
	\item the {\em boundary-to-boundary two-point function} $\lambda_{2,\partial,\Theta}\in\mathcal{D}^\prime(\partial M\times\partial M)$ such that $\lambda_{2,\partial,\Theta}\doteq(\iota_\partial^*\otimes\iota_\partial^*)\lambda_{2,\Theta}$ where $\iota_\partial:\partial M\to M$ is the embedding map of the boundary in $M$. 
\end{itemize}
Observe that $\lambda_{2,\partial,\Theta}$ is well-defined on account of Equation \eqref{eq: Hadamard_Wavefront_Set} and of \cite[Thm. 8.2.4]{Hor1}.
\end{rem}

The existence of Hadamard two-point functions is not a priori obvious and it represents an important question at the level of applications. Here we address it in two steps. First we focus on static, globally hyperbolic, asymptotically anti-de Sitter spacetimes and subsequently we drop the assumption that the underlying background is static, proving existence of Hadamard two-point functions via a deformation argument. 

Let us focus on the first step. To this end, on the one hand we need the boundary condition $\Theta$ to abide to Hypothesis \ref{hypothesis_WF_Theta}, while, on the other hand we make use of some auxiliary results from \cite{Wro17}, specialized to the case in hand. In the next statements it is understood that to any Hadamard two-point function $\lambda_{2,\Theta}$, it corresponds $\Lambda_\Theta: \dot{\mathcal{H}}_0^{-k,-\infty}(M) \rightarrow \mathcal{H}^{k,-\infty}_{loc}(M)$, with $k = \pm 1$. Recalling Definition \ref{Def: ellptic PsiDO} and \ref{Def: WF operatorial}, the following lemma holds true, {\it cf.} \cite[Lem. 5.3]{Wro17}:

\begin{lemma}
For any $q_1,q_2 \in {}^b S^*M$, $(q_1,q_2) \not \in WF^{Op}(\Lambda_\Theta)$ if and only if there exist neighbourhoods $\Gamma_i$ of $q_i$, $i=1,2$, such that for all $B_i \in \Psi_b^0(M)$ elliptic at $q_i$ satisfying $WF_b^{Op}(B_i) \subset \Gamma_i$, $B_1 \Lambda B_2 \in \mathcal{W}^{- \infty}_b(M)$.
\end{lemma}

Observe that this lemma entails in particular that, given any $f_i\in C^\infty(M)$, $i=1,2$ such that $\textrm{supp}(f_i)\subset\mathring{M}$ then $f_1 \Lambda_\Theta f_2$ has a smooth kernel over $\mathring{M}\times\mathring{M}$. In addition the following also holds true, {\it cf.} \cite[Prop. 5.6]{Wro17}:

\begin{prop}\label{prop:prop_wavefront_two_points}
Let $\Lambda_\Theta$ identify an Hadamard two-point function. If $(q_1,q_2) \in WF_b^{Op}(\Lambda_\Theta)$ for $q_1,q_2 \in T^*M \setminus \{0\}$, then $(q_1,q_1) \in WF_b^{Op}(\Lambda_\Theta)$ or $(q_2,q_2) \in WF_b^{Op}(\Lambda_\Theta)$.
\end{prop}

Given any two points $q_1$ and $q_2$ in the cosphere bundle ${}^bS^* M$, {\it cf.} Equation \eqref{Eq: cosphere bundle} we shall write $q_1 \dot{\sim} q_2$ if both $q_1$ and $q_2$ lie in the compressed characteristic bundle $\dot{\mathcal{N}}$ and they are connected by a generalized broken bicharacteristic, {\it cf.} Definition \ref{Def: generalized broken bicharacteristics}. With these data and using \cite[Prop. 5.9]{Wro17} together with Hypothesis \ref{hypothesis_WF_Theta} and with Theorems \ref{Thm: main theorem k positivo} and \ref{Thm: main theorem k negativo}, we can establish the following operator counterpart of the propagation of singularities theorem:

\begin{prop}\label{prop:operatorial_propagation}
Let $\Lambda_\Theta: \dot{\mathcal{H}}_0^{-1,-\infty}(M) \rightarrow \mathcal{H}^{1,-\infty}_{loc}(M)$ and suppose that $(q_1,q_2) \in WF^{Op}_b(\Lambda_\Theta)$. If $P_\Theta\Lambda_\Theta = 0$, then $q_1 \in \dot{\mathcal{N}}$ and $(q_1^\prime,q_2) \in WF_b^{Op}(\Lambda_\Theta)$ for every $q_1^\prime$ such that $q_1^\prime \dot{\sim} q_1$. Similarly, if $\Lambda_\Theta P_\Theta  = 0$, then $q_2 \in \dot{\mathcal{N}}$ and $(q_1,q_2^\prime) \in WF_b^{Op}(\Lambda_\Theta)$ for all $q_2^\prime$ such that $q_2^\prime \dot{\sim} q_2$.
\end{prop}

\noindent Our next step consists of refining Theorem \ref{thm:wavefront_set_propagator} in $\mathring{M}$, {\it cf.} for similarities with \cite[Cor. 4.5]{DF18}. 

\begin{corollary}
Let $G_\Theta : \mathcal{H}^{-1,-\infty}(\mathring{M})\rightarrow \mathcal{H}^{1,-\infty}(\mathring{M})$ be the restriction to $\mathring{M}$ of the causal propagator as per Remark \ref{Rem: causal propagator} . Then
$$ WF_b^{Op}(G_\Theta) = \{ (q_1,q_2) \in {}^b S^*\mathring{M} \times {}^b S^*\mathring{M} \ | \ q_1 \dot{\sim} q_2 \}.$$
\end{corollary}

\begin{proof}
A direct application of Theorem \ref{thm:wavefront_set_propagator} yields 
\begin{equation*}
WF^{Op}(G_\Theta) \subseteq \{(q_1,q_2) \in {}^b S^*\mathring{M} \times {}^b S^*\mathring{M} \; | \; q_1 \dot{\sim} q_2 \}
\end{equation*}
From this inclusion, it descends that every pair of points in the singular support of $G$ is connected by a generalized broken bicharacteristic completely contained in $\mathring{M}$. Since $^{b}T^*\mathring{M} \simeq T^*\mathring{M}$, we can apply \cite[Ch.4, Thm. 16]{BF09} and the sought statement is proven.
\end{proof}

With these data, we are ready to address the main question of this section. Suppose that $(M,g)$ is a static, globally hyperbolic, asymptotically AdS spacetime, {\it cf.} Definition \ref{Def: asymptotically AdS} and \ref{Def: standard static}. Let $P_\Theta$ be the Klein-Gordon operator as per Equation \eqref{Eq: PTheta} and let $\Theta\equiv\Theta_K$ be a static boundary condition as per Theorem \ref{Theorem: construction of the advanced and retarded propagators}. For simplicity we also assume that the spectrum of $E_K$ is contained in the positive real axis. Then the following key result holds true: 

\begin{prop}\label{prop:wavefront_bulk_to_bulk}
Let $(M,g)$ be a static, globally hyperbolic asymptotically AdS spacetime and let $P_{\Theta_K}$ be the Klein-Gordon operator with a static and physically admissible boundary condition as per Definition \ref{Def: admissible boundary conditions} Then there exists a Hadamard two-point function associated to $P_\Theta$,  $\lambda_{2,\Theta_K}\in\mathcal{D}^\prime(M\times M)$ such that, for all $f_1,f_2\in\mathcal{D}(M)$ 
\begin{align}\label{Eq: Hadamard 2-pt static}
	\lambda_{2,\Theta_K}(f_1,f_2)\doteq 2i\int_{\mathbb{R}^2}\textrm{d}t\textrm{d}t'\,\bigg( f_1(t)\bigg|\frac{\exp\big[i E_{\Theta_K}^{\frac{1}{2}}(t-t^\prime)\big]}{(-E_{\Theta_K})^{\frac{1}{2}}}f_2(t^\prime)\bigg),
\end{align}
\end{prop}

\begin{proof}
Observe that, per construction $\lambda_{2,\Theta_k}$ is a bi-solution of the Klein-Gordon equation associated to the operator $P_{\Theta_K}$ and it abides to Equation \eqref{Eq: constraints on 2-pt function}. We need to show that Equation \eqref{eq: Hadamard_Wavefront_Set} holds true. To this end it suffices to combine the following results. From \cite{Sahlmann:2000fh} one can infer that, the restriction of $\mathring{\lambda}_{2,\Theta_K}$, the bulk-to-bulk two-point distribution, to every globally hyperbolic submanifold of $M$ not intersecting the boundary is consistent with Equation \eqref{eq: Hadamard_Wavefront_Set}. At this point it suffices to invoke Proposition \ref{prop:prop_wavefront_two_points} and \ref{prop:wavefront_bulk_to_bulk} to draw the sought conclusion.
\end{proof}

\begin{rem}
Observe that, from a physical viewpoint, in the preceding theorem, we have individuated the two-point function of the so-called {\em ground state} with boundary condition prescribed by $\Theta_K$.
\end{rem}

\subsection{A Deformation Argument}\label{Sec: A deformation Argument}
In order to prove the existence of Hadamard two-point functions on a generic asymptotically anti-de Sitter spacetime for a Klein-Gordon field with prescribed static boundary condition, we shall employ a a deformation argument akin to that first outlined in \cite{FNW81} on globally hyperbolic spacetimes with empty boundary. 

To this end we need the following lemma, see \cite[Lem. 4.6]{Wro17}, slightly adapted to the case in hand. In anticipation, recalling Equation \eqref{Eq: metric near the boundary}, we say that a globally hyperbolic, asymptotically AdS spacetime is {\em even} modulo $\mathcal{O}(x^3)$ close to $\partial M$ if $h(x)=h_0+x^2h_1(x)$ where $h_1$ is a symmetric two-tensor, see \cite[Def. 4.3]{Wro17}.

\begin{lemma}\label{Lem: deformation_spacetime}
Suppose $(M,g)$ is a globally hyperbolic, asymptotically anti-de Sitter spacetime. For any $\tau_2 \in \mathbb{R}$ there a static, globally hyperbolic asymptotically AdS spacetime $(M,g^\prime)$ as well as $\tau_0,\tau_1$ with $\tau_0<\tau_1<\tau_2$ such that $g^\prime=g$ if $\{ \tau \geq \tau_1\}$, while, if $\{ \tau \leq \tau_0\}$, $(M,g^\prime)$ is isometric to a standard static asymptotically AdS spacetime $(M,g_S)$ which is even modulo $\mathcal{O}(x^3)$ and in which $C \leq \beta \leq C^{-1}$ for some $C>0$, with $\beta$ as in Equation \eqref{Eq: line element}.
\end{lemma}

Consider now a generic, globally hyperbolic, asymptotically anti-de Sitter spacetime $(M,g)$ and a deformation as per Lemma \ref{Lem: deformation_spacetime}. Observe that, per construction, all generalized broken bicharacteristics reach the region of $M$ with $\tau \in [\tau_1,\tau_2]$. This observation leads to the following result which is a direct consequence of the propagation of singularities theorem \ref{Thm: main theorem k negativo} and \ref{Thm: main theorem k positivo}. Mutatis mutandis, the proof is as that of \cite[Lem. 5.10]{Wro17} and, thus, we omit it.

\begin{lemma}\label{Lem: deformation_hadamard}
Suppose that $\Lambda_\Theta\in\mathcal{D}^\prime(M\times M)$ is a bi-solution of the Klein-Gordon equation ruled by $P_\Theta$ abiding to Equation \eqref{Eq: constraints on 2-pt function} and with a wavefront set of Hadamard form in the region of $M$ such that $\tau_1<\tau<\tau_2$. Then $\Lambda_\Theta$ is a Hadamard two-point function.
\end{lemma}

\noindent To conclude, employing Corollary \ref{Cor: time_slice_axiom} we can prove the sought result:

\begin{thm}
Let $(M,g)$ be a globally hyperbolic, asymptotically anti-de Sitter spacetime and let $(M_S,g_S)$ be its static deformation as per Lemma \ref{Lem: deformation_spacetime}. Let $\Theta_K$ be a static and physically admissible boundary condition so that the Klein-Gordon operator $P_{\Theta_K}$ on $(M_S,g_S)$ admits a Hadamard two-point function as per Proposition \ref{prop:wavefront_bulk_to_bulk}. Then there exists a Hadamard two point-function on $(M,g)$ for the associated Klein-Gordon operator with boundary condition ruled by $\Theta_K$.
\end{thm}

\begin{proof}
Let $(M,g)$ be as per hypothesis and let $(M,g_S)$ be a static, globally hyperbolic, asymptotically AdS spacetime such that there exists a third, globally hyperbolic, asymptotically AdS spacetime $(M,g^\prime)$ interpolating between $(M,g)$ and $(M,g_S)$ in the sense of Lemma \ref{Lem: deformation_spacetime}. On account of Theorem \ref{Thm: globally hyperbolic}, in all three cases $M$ is isometric to $\mathbb{R}\times\Sigma$. 

On account of Proposition \ref{prop:wavefront_bulk_to_bulk}, on $(M,g_S)$ we can identify an Hadamard two-point function as in Equation \eqref{Eq: Hadamard 2-pt static} subordinated to the boundary condition $\Theta_K$. We indicate it with $\lambda_{2,S}$ omitting any reference to $\Theta_K$ since it plays no explicit r\^{o}le in the analysis. 

Focusing the attention on $(M,g^\prime)$, Lemma \ref{Lem: deformation_spacetime} guarantees that, if $\tau<\tau_0$, $\tau$ being the time coordinate along $\mathbb{R}$, then therein $(M,g^\prime)$ is isometric to $(M,g_S)$. Calling this region $M_0$, the restriction $\lambda_{2,S}|_{M_0\times M_0}$ identifies a two-point distribution of Hadamard form. Notice that we have omitted to write explicitly the underlying isometries for simplicity of notation.

Using the time-slice axiom in Corollary \ref{Cor: time_slice_axiom}, for any pair of test-functions $f,f^\prime\in\mathcal{D}(M^\prime)$ such that for all $p\in\textrm{supp}(f)\cup\textrm{supp}(f^\prime)$, $\tau(p)>\tau_0$, we set $h=P_{\Theta_K}\chi G_{\Theta_K}(f)$ and $h^\prime=P_{\Theta_K}\chi G_{\Theta_K}(f^\prime)$ where $G_{\Theta_K}$ is the causal propagator associated to $P_{\Theta_K}$ in $(M,g^\prime)$, while $\chi=\chi(\tau)$ is any smooth function such that there exists $\tau_1,\tau_2<\tau_0$ for which $\tau=0$ if $\tau<\tau_1$ while $\chi=1$ if $\tau>\tau_2$. We define
$$\lambda^\prime_2(f,f^\prime)=\lambda_{2,S}(h,h^\prime).$$
Observe that $h,h^\prime\in C^\infty_{tc}(M)$ and therefore the right-hand side of this identity is well-defined. In addition, since $G_{\Theta_K}$ is continuous on $\mathcal{D}(M)$, sequential continuity entails that $\lambda_2^\prime\in\mathcal{D}(M^\prime\times M^\prime)$. In addition, per construction, it is a solution of the Klein-Gordon equation ruled by $P_{\Theta_K}$ on $(M^\prime,g^\prime)$ and abiding to Equation \eqref{Eq: constraints on 2-pt function}. 

Furthermore Lemma \ref{Lem: deformation_hadamard} yields that $\lambda_2^\prime$ is of Hadamard form. 

To conclude it suffices to focus on $(M,g)$ recalling that there exists $\tau_1\in\mathbb{R}$ such that, in the region $(M_1,g^\prime)\subset (M,g^\prime)$ for which $\tau>\tau_1$, $(M,g^\prime)$ is isometric to $(M,g)$. Hence,we can repeat the argument given above. More precisely we consider $\lambda_2^\prime|_{M^\prime\times M^\prime}$ and, using the time-slice axiom, see Corollary \ref{Cor: time_slice_axiom}, we can identify $\lambda_2\in\mathcal{D}^\prime(M\times M)$ which is a solution of the Klein-Gordon equation ruled by $P_{\Theta_K}$ and it abides to Equation \eqref{Eq: constraints on 2-pt function}. Lemma \ref{Lem: deformation_hadamard} entails also that it is of Hadamard form, hence proving the sought result.
\end{proof}

\section*{Acknowledgments}

We are grateful to Benito Juarez Aubry for the useful discussions which inspired the beginning of this project and to Nicolò Drago both for the useful discussions and for pointing out references \cite{Ginoux, GMP}. We are also grateful to Simone Murro and to Micha\l Wrochna for the useful discussions. The work of A. Marta is supported by a fellowship of the Università Statale di Milano, which is gratefully acknowledged. C. Dappiaggi is grateful to the Department of Mathematics of the Università Statale di Milano for the kind hospitality during the realization of part of this work.

\end{document}